\documentclass[11pt]{amsart}

\usepackage{amsmath, amsthm, amssymb, amsfonts, enumerate}
\usepackage[colorlinks=true,linkcolor=blue,urlcolor=blue]{hyperref}
\usepackage{dsfont}
\usepackage{color}
\usepackage{geometry}
\usepackage{todonotes}

\usepackage{graphicx}
\usepackage{caption}
\usepackage{subcaption}

\def \cC{\mathcal{C}}
\def \cB{\mathcal{B}}
\def \cE{\mathcal{E}}
\def \cS{\mathcal{S}}

\def \cF{\mathcal{F}}
\def \cG{\mathcal{G}}
\def \cL{\mathcal{L}}
\def \cM{\mathcal{M}}

\def \cT{\mathcal{T}}
\def \T{T}
\def \prob{\mathsf P}
\def \Q{\mathsf Q}
\def \pprob{\widetilde{\mathsf{P}}}
\def \media{\mathsf E}
\def \mmedia{\widetilde{\mathsf{E}}}
\def \realip{\mathbb{R}_+}
\def \reali{\mathbb{R}}
\def \UU{\widehat{U}}
\def \VV{\widehat{V}}

\newcommand{\vs}{\vspace}
\newcommand{\eps}{\varepsilon}

\newtheorem{theorem}{Theorem}[section]
\newtheorem{lemma}[theorem]{Lemma}
\newtheorem{corollary}[theorem]{Corollary}
\newtheorem{proposition}[theorem]{Proposition}

\newtheorem{remark}[theorem]{Remark}
\newtheorem{Assumptions}[theorem]{Assumption}

\DeclareMathOperator*{\esssup}{ess\,sup}

\title[On the free boundary of an annuity purchase]{On the free boundary of an annuity purchase}
\author[T.~De Angelis and G.~Stabile]{Tiziano De Angelis \and Gabriele Stabile}
\address{T.~De Angelis: School of Mathematics, University of Leeds, Woodhouse Lane, LS2 9JT Leeds, UK.}
\email{\href{mailto:t.deangelis@leeds.ac.uk}{t.deangelis@leeds.ac.uk}}
\address{G.~Stabile: Dipartimento di Metodi e Modelli per l'Economia, il Territorio e la Finanza, Sapienza-Universit\`{a} di Roma, Roma, Italy}
\email{\href{mailto:gabriele.stabile@uniroma1.it}{gabriele.stabile@uniroma1.it}}
\date{\today}

\numberwithin{equation}{section}

\begin{document}

\begin{abstract}
It is known that the decision to purchase an annuity may be associated to an optimal stopping problem. However, little is known about optimal strategies, if the mortality force is a generic function of time and if the \emph{subjective} life expectancy of the investor differs from the \emph{objective} one adopted by insurance companies to price annuities. In this paper we address this problem considering an individual who invests in a fund and has the option to convert the fund's value into an annuity at any time. We formulate the problem as a real option and perform a detailed probabilistic study of the optimal stopping boundary. Due to the generic time-dependence of the mortality force, our optimal stopping problem requires new solution methods to deal with non-monotonic optimal boundaries.
\end{abstract}

\maketitle

\section{Introduction}
In an ageing world an accurate management of retirement wealth is crucial for financial well being.
It is important for working individuals to carefully consider the existing offer of financial and insurance products designed for retirement, beyond the state pension. This offer includes for example occupational pension funds and tax-advantaged retirement accounts (e.g.~Individual Retirement Account (US)).
Most of these products rely on annuities to turn retirement wealth into guaranteed lifetime retirement income.
Life annuities provide a lifelong stream of guaranteed income in exchange for a (single or periodic) premium.
The purchase of an annuity helps individuals to manage the longevity risk, i.e.~the risk of outliving their financial wealth, but it is usually an irreversible transaction.
In fact, most annuity contracts impose steep penalties for partial or complete cancellation by the policyholder, especially in the early years of the contract.

Timing an annuity purchase (so-called annuitization) is a complex financial decision that depends on several risk factors as, e.g., market risk, longevity risk, potential future need of liquid funds and bequest motive.
The study of this topic has motivated a whole research field since the seminal contribution of Yaari \cite{Yaari}, who showed that individuals with no bequest motive should convert all their retirement wealth into annuities.

After Yaari, several authors have analysed the annuitization decision under the so-called \textit{all or nothing} institutional arrangement, where a lifetime annuity is purchased in a single transaction (as opposed to gradual annuitization). Initially an individual's wealth is invested in the financial market and, at the time of an annuity purchase, it is converted into a lifetime annuity. The central idea in this literature is to compare the value deriving from an immediate annuitization with the value of deferring it, while investing in the financial market. Therefore, a strict analogy holds with the problem of exercising an American option, and the annuitization decision can be considered as the exercise of a real option.

Milevsky \cite{M98} proposed a model where an individual defers annuitization for as long as the financial investment's returns guarantee a consumption flow which is at least equal to the one provided by the annuity payments. In particular \cite{M98} adopts a criterion based on controlling the probability of a consumption shortfall.

Other papers study the optimal annuitization time in the context of utility maximisation, and formulate the problem as one of optimal stopping and control. The investor aims at maximising the expected utility of consumption (pre-retirement) and of annuity payments (post-retirement).

Assuming \emph{constant force of mortality} and CRRA utility, Stabile \cite{Stabile} analytically solves a time homogeneous optimal stopping problem. He proves that if the individual has the same degree of risk aversion before and after the annuitization, then an annuity is purchased either immediately or never (the so-called \textit{now or never} policy). Instead, in case the individual is more risk averse during the annuity payout phase, the annuity is purchased as soon as the wealth falls below a constant threshold (the \emph{optimal stopping boundary}).

Constant force of mortality is also assumed in Gerrard et al.~\cite{GHV} and Liang et al.~\cite{LPG}. The model in \cite{GHV} is analogous to the one studied in \cite{Stabile}, but with quadratic utility functions, and authors find a closed-form solution: if $(X_t)_{t\ge0}$ represents the individual's wealth process, then it is optimal to stop when $X$ leaves a specific interval (hence the optimal stopping boundary is formed by the endpoints of such interval). In \cite{LPG}, in contrast to the previous papers, the authors assume that the individual may continue to invest and consume after annuitization. By using martingale methods, explicit solutions are provided in the case of CRRA utility functions. Contrarily to \cite{GHV}, in \cite{LPG} the optimal annuitization occurs when the wealth process \emph{enters} a specific interval, whose endpoints form the optimal stopping boundary.

Assuming \emph{time-dependent force of mortality}, Milevsky and Young \cite{MY} analyze both the \textit{all or nothing} market and the more general \textit{anything anytime} market, where gradual annuitization strategies are allowed. For the \textit{all or nothing} market they find that the optimal annuitization time is deterministic as an artifact of CRRA utility. Thus, the annuitization decision is independent from the individual's wealth.

Our work is more closely related to work by Hainaut and Deelstra \cite{HD}. They consider an individual whose retirement wealth is invested in a financial fund which eventually must be converted into an annuity. The fund is modelled by a jump diffusion process and pays dividends at a constant rate. The mortality force is a time-dependent, deterministic function and the individual aims at maximising the market value of future cashflows before and after annuitization. According to the insurance practice, it is assumed that the individual can only purchase the annuity by a given maximal age. Authors in \cite{HD} cast the problem as an optimal stopping one and write a variational inequality for the value function. They then use Wiener-Hopf factorisation and a time stepping method to solve the variational problem numerically. Hainaut and Deelstra argue that the decision to purchase the annuity should be triggered by either an upper or a lower, time-dependent threshold in the time-wealth plane. Here and in what follows $t$ is the state variable associated to time and $x$ is the one associated to wealth, so that the time-wealth plane is simply referred to as the $(t,x)$-plane. The threshold discussed in \cite{HD} is the optimal stopping boundary in their setting and we denote it by a mapping $t\mapsto b(t)$. Numerical examples are provided in \cite{HD} where the annuitization occurs when the value of the financial fund is high enough or, in alternative, low enough.

In this paper we perform a detailed mathematical study of the optimal stopping problem associated to an annuitization decision similar to that considered in \cite{HD}. In the interest of a rigorous analysis of the optimal stopping boundary we simplify the dynamic of the financial fund by considering a geometric Brownian motion with no jumps. As in \cite{HD} we look at the maximisation of future expected cashflows for an individual who joins the fund and has the opportunity to purchase an annuity on a time horizon $[0,T]$. Time $0$ is the time when the individual joins the fund and time $T$ is the time by which the individual reaches the maximal age for an annuity purchase. The present value of future expected cashflows, evaluated at the optimum, gives us the so-called value function $V$.

Notice that, a closer inspection of the problem formulation in \eqref{value1b} below, shows that at time $T$ the fund is converted into an annuity (the same occurs in \cite{HD}). This means that the individual will eventually purchase the annuity at time $T$, but she also has an option to buy it earlier. One could think of this feature as part of the fund's contract specifications or as commitment of the investor at time $0$. It is however important to remark that the methods developed in this paper apply also to the case $T=+\infty$, up to some minor changes (see also Remark \ref{rem:T} for further details).

One of the key features of the model presented here is the use of a rather general time-dependent, deterministic mortality force. This is a realistic assumption commonly made in the actuarial profession. As in \cite{MY}, we consider two different mortality forces: a \emph{subjective} one, used by the individual to weigh the future cashflows (denoted $\mu^S$), and an \emph{objective} one, used by the insurance company to price the annuity (denoted $\mu^O$). The interplay between these two different mortality forces contributes to some key qualitative aspects of the optimal annuitization decision (see Section \ref{numerics} for more details).
Interestingly the generic time-dependent structure of the mortality force constitutes also the major technical challenge in the mathematical study of the problem.

On the one hand, standard optimal stopping results ensure that the time-wealth plane splits into a continuation region $\cC$, where the option to wait has strictly positive value, and a stopping region $\cS$, where the annuity should be immediately purchased. Denoting by $(X_t)_{t\ge0}$ the process that represents the fund's value (or equally the individual's retirement wealth), an optimal stopping rule is given by stopping at the first time the two-dimensional process $(t,X_t)_{t\ge0}$ enters the set $\cS$. Moreover, under some mild technical assumptions these two sets are split by an optimal boundary (free boundary, in the language of PDE), which only depends on time, i.e.~$t\mapsto b(t)$.

On the other hand, technical difficulties arise when trying to infer properties of the boundary $b$. In fact, due to the generic time dependence of the mortality force, it is not possible to establish any monotonicity of the mapping $t\mapsto b(t)$. It is well known in optimal stopping and free boundary theory that monotonicity of $b$ is the key to a rigorous study of the regularity of the boundary (e.g.~continuity) and of the value function (e.g.~continuous differentiability). The interested reader may consult \cite{PS}, for a collection of examples, and the introduction in \cite{DAS}, for a deeper discussion.

We overcome this major technical hurdle by proving that the optimal boundary is in fact a locally Lipschitz continuous function of time. In order to achieve this goal we rely only on probabilistic methods which are new and specifically designed to tackle our problem. This approach draws from similar ideas in \cite{DAS}, but we emphasise that our problem falls outside the class of problems addressed in that paper (see the discussion prior to Theorem \ref{th1} below).

Once Lipschitz regularity is proven we then obtain also that the value function $V$ is continuously differentiable in $t$ and $x$, at all points of the $(t,x)$-plane and in particular across the boundary of $\cC$. This is a stronger result than the more usual smooth-fit condition, which states that $z\mapsto V_{x}(t,z)$ is continuous across the optimal boundary. Finally, we find non-linear integral equations that characterize uniquely the free boundary and the value function.

The analysis in this paper is completed by solving numerically the integral equation for some specific examples and studying their sensitivity to variations in model's parameters. It is important to remark that the optimal boundary turns out to be non monotonic in some of our examples, under natural assumptions on the parameters. This fact shows that the new approach developed in this paper is indeed really necessary to study the annuitization problem.

In summary our contribution is at least twofold. On the one hand, we add to the literature concerning annuitization problems, in the \emph{all-or-nothing} framework, by addressing models with time-dependent mortality force. As we have discussed above, and to the best of our knowledge, such models were only considered in \cite{MY} (which produces only deterministic optimal strategies) and in \cite{HD} (mostly in a numerical way). We provide a rigorous theoretical analysis of the optimal annuitization strategy, in terms of the optimal boundary $b$. Our study also reveals behaviours not captured by \cite{HD} as, e.g.,~lack of monotonicity of $b$. The latter may reflect the change over time in the investor's priorities, due to (deterministic) variations in the mortality force.
On the other hand, it is rather remarkable that we started by considering an applied problem, with a somewhat canonical and seemingly innocuous formulation, but we soon realised that its rigorous analysis is far from being trivial. Therefore we developed methods which are new in the probabilistic literature on optimal stopping and of independent interest.

Finally, in order to relate our work to the PDE literature in this area, it may be worth noticing that \cite{FS02} (and later \cite{C06}) studies a free boundary problem motivated by optimal retirement. In that paper an investor can decide to retire earlier than a given terminal time $T$. Early retirement benefits are defined by a function $\Psi(t,s)$ of time and of the current salary $s$. The problem is addressed exclusively with variational inequalities and the free boundary depends on time since $t\mapsto \Psi(t,s)$ increases linearly. However, contrarily to our model, in \cite{FS02} and \cite{C06} the mortality force is assumed constant.

The rest of the paper is organized as follows. In Section 2 we introduce the financial and actuarial assumptions and then the optimal annuitization problem. In Section 3 we provide some continuity properties of the value function and useful probabilistic bounds on its gradient. In Section 4 we illustrate sufficient conditions under which the shape of the continuation and stopping regions can be established, and we study the regularity of the optimal boundary. Moreover, we find non-linear integral equations that characterise uniquely the free boundary and the value function. In Section 5 we present some numerical examples to illustrate the range of applicability of our assumptions. In Section 6 we provide some final remarks and extensions.

\section{Problem formulation}

In our model we consider an individual (or investor) and an insurance company who are faced with two sources of randomness: a \emph{financial market} and the \emph{survival probability} of the individual. We assume that the two sources are independent and we also assume that the individual and the insurance company have different beliefs about the demographic risk, while they share the same views on the financial market. It is therefore convenient to construct initially two probability spaces: one that models the financial market and another one that models the demographic component. The time horizon of the problem is fixed and it is denoted by $T<+\infty$.

\subsection{Financial and demographic models} For the financial market we consider a complete probability space $(\Omega,\cF,\prob)$ carrying a 1-dimensional Brownian motion $(B_t)_{t\ge0}$. The filtration generated by $B$ is denoted by $(\cF_t)_{t\ge0}$ and it is augmented with $\prob$-null sets. The portion of the individual's wealth allocated for an annuity purchase, and invested in a financial fund\footnote{In what follows we make no distinction between the fund's dynamics and the wealth's dynamics.}, prior to the annuitization is modelled by a stochastic process $(X_t)_{t\ge0}$. Its dynamic reads
\begin{align}
dX_t^{x} = (\theta-\alpha)X_t^{x} dt+\sigma X_t^{x} dB_t, \qquad X_0^{x} = x> 0,
\label{wealth}
\end{align}
where $\theta$ is the average continuous return of the financial investment, $\alpha$ is the constant dividend rate and $\sigma>0$ is the volatility coefficient. 

For the demographic risk we consider another probability space: given a measurable space $(\Omega',\cF')$ we let $\Q^S$ and $\Q^O$ denote two probability measures on $(\Omega',\cF')$ and assume that $(\Omega',\cF',\Q^i)$, $i=S,O$ are both complete. The measure $\Q^S$ is associated with the \emph{subjective} survival probability of the individual. On the contrary, $\Q^O$ refers to the \emph{objective} survival probability used by the insurance company to price annuities and it is public information.

The individual is aged $\eta>0$ at time zero in our problem and this value is given and fixed throughout the paper. We define a random variable $\Gamma_D:(\Omega',\cF')\to(\mathbb{R}_+,\cB(\mathbb{R}_+))$ that represents the time of death of the individual and, for $i=S,O$ we define the hazard functions
\[
_s p^i_{\eta+t}:=\Q^i(\Gamma_D>\eta+t+s\,|\,\Gamma_D>\eta+t)
\]
with $s,t\ge 0$. These represent the subjective/objective probability that an individual who is alive at the age $\eta+t$ will survive to the age $\eta+t+s$ (we follow the standard actuarial notation for $_s p^i_{\eta+t}$).
Let $\mu^S:[0,+\infty)\to[0,+\infty)$ and $\mu^O:[0,+\infty)\to[0,+\infty)$ be deterministic functions, representing the \emph{subjective} and \emph{objective} mortality force, respectively. Then, for $i=S,O$ we have
\begin{align}\label{pS}
_s p^i_{\eta+t}=\exp \left( -\int_{0}^{s}\mu^i(\eta+t+u)du\right),\quad\text{for $t,s\ge0$}.
\end{align}
The different survival probability functions adopted by the insurer and the individual account for the imperfect information available to the insurer on the individual's risk profile.

Finally, we say that $\cM^S:=(\Omega\times\Omega',\cF\otimes\cF',\prob\times\Q^S)$ is the probability space for the individual and, for completeness, that $\cM^O:=(\Omega\times\Omega',\cF\otimes\cF',\prob\times\Q^O)$ is the probability space for the insurance company.
\begin{remark}\label{rem:detmu}
Functions $\mu^S$ and $\mu^O$ are given at outset and are not updated during the optimisation. Updating in a non trivial way would require the use of a stochastic dynamic for the mortality force which, in general, would lead to a more complex problem.
\end{remark}

\subsection{The optimisation problem}
The insurance company uses its probabilistic model, based on objective survival probabilities, to price annuities. In particular, according to standard actuarial theory, the value at time $t>0$ of a life annuity that is payable continuously at a rate of one monetary unit per year (purchased by the individual aged $\eta+t$) is given by
\[
a^O_{\eta+t}= \int_{0}^{\infty} e^{-\widehat \rho u} {}_u p^O_{\eta+t} du.
\]
Here $\widehat \rho>0$ is a constant interest rate guaranteed by the insurer.

In our model the fund is automatically converted into an annuity at time $T$ but the individual has the option to annuitize prior to $T$. If she decides to annuitize at a time $t\in[0,T]$, with the fund's value equal to $X$, then the annuity payout rate is constant and it reads
\begin{equation}\label{eq:P}
P_{\eta+t}=\frac{X-K}{a^O_{\eta+t}},
\end{equation}
where the constant $K$ is either a fixed acquisition fee ($K>0$) or a tax incentive ($K<0$). The case $K=0$ leads to trivial solutions as explained in Remark \ref{rem:K0} below.
From the modelling point of view $T<+\infty$ reflects the fact that insurance companies typically have maximum age limit for the purchase of an annuity (this is noticed also in \cite{HD}).

The optimization criterion pursued by the individual is the maximization of the present value of future expected cash-flows, via the optimal timing of the annuity purchase under the model $\cM^S$. Letting $\cE^S[\cdot]$ be the expectation under the measure $\prob\times\Q^S$, if the individual is alive at time $t$ the optimisation problem reads
\begin{align}
\label{value1b}
V_t=\esssup_{\tau \in \cT_{t,T} }\cE^S\bigg[&\! \int_{t}^{\Gamma_D \wedge \tau }\!\!e^{-\rho s}\alpha X_s ds+\mathds{1}_{\{\Gamma_D \leq \tau \}}e^{-\rho \Gamma_D} X_{\Gamma_D}\\
&\hspace{+40pt} +P_{\eta+\tau}\hspace{-3pt}\int_{\Gamma_D \wedge \tau}^{\Gamma_D}e^{-\rho s}ds\Big|\cF_t\cap\{\Gamma_D>t\} \bigg]\notag
\end{align}
where $\cT_{t,T}$ is the set of $(\cF_s)_{s\ge0}$-stopping times taking their values in $[t,T]$ and $\rho>0$ is a discount rate.
Before annuitization, i.e.~for $s<\tau$, the individual receives dividends from the fund at rate $\alpha$. After annuitization, i.e.~for $s>\tau$, she gets the continuous annuity payment at constant annual rate $P_{\eta+\tau}$. In case the individual dies before the time of the annuity purchase, i.e.~on the event $\{\Gamma_D \leq \tau\}$, she leaves a bequest equal to her wealth.

\begin{remark}\label{rem:CG}
Thanks to a result in \cite{CG12} we show in appendix that there is no loss of generality in using stopping times from $\cT_{t,T}$. That is, we obtain the same value in \eqref{value1b} as if we were using stopping times of the enlarged filtration $(\cG_t)_{t\ge 0}$ where $\cG_t=\cF_t\vee\sigma(\{\Gamma_D>s\}\,,\,0\le s\le t)$.
\end{remark}

Due to the assumed independence between the demographic uncertainty and fund's returns (i.e.\ $\Gamma_D$ independent of $(B_t)_{t\ge0}$), and since the optimisation is over $(\cF_t)_{t\ge0}$-stopping times, the value function can be rewritten by using Fubini's theorem and \eqref{pS} as
\begin{align}
\label{value2b}
V_t=\esssup_{\tau \in \cT_{t,T} }\media \bigg[\! \int_{t}^{\tau}\!\!e^{-\int_{t}^{s}r(u) du}\beta(s) X_s ds+e^{-\int_{t}^{\tau}r(u)du}G(\tau,X_\tau)\Big|\cF_t\bigg]
\end{align}
where $\media[\,\cdot\,]$ is the expectation under $\prob$, $r(t):=\rho+\mu^S(\eta+t)$, $\beta(t):=\alpha+\mu^S(\eta+t)$, $G(t,x)=f(t)(x-K)$ and
\begin{align}
\label{f} f(t)=\frac{a^S_{\eta+t}}{a^O_{\eta+t}}.
\end{align}
Here $a^S_{\eta+t}$ is the individual's subjective valuation of the annuity, i.e.
\[
a^S_{\eta+t}=\int_{0}^{\infty} e^{- \rho u} {}_u p^S_{\eta+t} du.
\]
The function $f(\cdot)$ in \eqref{f} is the so-called ``money's worth''.

Since we are in a Markovian setting we have $\media[\,\cdot\,|\cF_t]=\media[\,\cdot\,|X_t]$. In particular if $X_t=x>0$, $\prob$-a.s.~we find it convenient to use the notation $\media_{t,x}[\,\cdot\,]:=\media[\,\cdot\,|X_t=x]=\media[\,\cdot\,|\cF_t]$. Moreover, the process $X$ is time-homogeneous, so that
\[
\mathsf{Law}\big((s,X_s)_{s\ge t}\big|X_t=x\big)=\mathsf{Law}\big((t+s,X_s)_{s\ge 0}\big|X_0=x\big).
\]
Using the above notations, for any given $(t,x)\in[0,T]\times(0,+\infty)$ we can rewrite \eqref{value2b} as
\begin{align}
\label{value2}
V(t,x)=\sup_{0 \leq \tau \leq T-t} \media \bigg[\! \int_{0}^{\tau}\!\!e^{-\int_{0}^{s}r(t+u) du}\beta(t+s) X_s^x ds+e^{-\int_{0}^{\tau}r(t+u)du}G(t+\tau,X^x_\tau)\bigg],
\end{align}
where we also say that $s_1\le\tau\le s_2\iff\tau\in\cT_{s_1,s_2}$ (this should cause no confusion because all stopping times in this paper belong to $\cT_{s_1,s_2}$ for some $s_1\le s_2$).

\subsection{The variational problem}

Before closing this section we introduce the variational problem naturally associated to \eqref{value2}.
Let $\mathcal{L}$ be the second order differential
operator associated to the diffusion \eqref{wealth}, i.e.
\[
(\mathcal{L}F)(x)=(\theta-\alpha)xF_x(x)+\tfrac{\sigma^2 x^2}{2} F_{xx}(x),\quad\text{for $F\in C^{2}(\realip)$}.
\]
Assuming for a moment that $V$ is regular enough, by applying the dynamic programming principle and It\^o's formula, we expect that the value function should satisfy the following variational inequality
\begin{equation}
\label{HJB}
\textrm{max} \Big\{ (V_t+\mathcal{L}V-r(\cdot)V)(t,x)+\beta(t)x, G(t,x)-V(t,x) \Big\}= 0, \ \ \ \ \ (t,x)\in(0,T)\times\realip
\end{equation}
with terminal condition $V(T,x)=G(T,x)$, $x\in\realip$. In the rest of the paper we will show that \eqref{HJB} holds in the a.e.\ sense with $V\in C^1([0,T)\times\realip)\cap C([0,T]\times\realip)$ and $V_{xx}\in L^\infty_{loc}([0,T)\times\realip)$. Moreover we will study the geometry of the set where $V=G$, i.e.\ the so-called stopping region.

\section{Properties of the value function}\label{sec:value}
In this section we provide some continuity properties of the value function and useful probabilistic bounds on its gradient. In what follows, given a set $A\subseteq[0,T]\times\realip$ we will sometimes denote $A\cap\{t<T\}:=A\cap([0,T)\times\realip)$. Also, we make the next standing assumption in the rest of the paper.
\begin{Assumptions}
\label{ass1}
$\mu^S(\cdot)$ and $\mu^O(\cdot)$ are continuously differentiable on $[0,+\infty)$.
\end{Assumptions}
To study the optimization problem \eqref{value2} we find it convenient to introduce the function
\begin{equation}
\label{w}
v(t,x)= V(t,x)-G(t,x),\quad\text{$(t,x)\in[0,T]\times\realip$}
\end{equation}
which may be financially understood as the value of the option to delay the annuity purchase.

We can easily compute
\begin{align}\label{HG}
H(t,x):=(G_t+\mathcal{L} G-r(\cdot)G)(t,x)+\beta(t)x=g(t)x+K \ell(t),
\end{align}
where
\begin{align}
\label{g}g(t):=f'(t)+\beta(t)(1-f(t))+(\theta-\rho)f(t)\:\:\:\text{and}\:\:\: \ell(t):=r(t)f(t)-f'(t).
\end{align}
An application of It\^o's formula gives
\begin{align*}
\media & \bigg[ e^{-\int_{0}^{\tau}r(t+u)du}G(t+\tau,X_\tau^x)\bigg]\\
=&G(t,x)+\media \bigg[ \int_{0}^{\tau}e^{-\int_{0}^{s}r(t+u)du}\Big(H(t+s,X^x_s)-\beta(t+s)X^x_s\Big) ds\bigg]
\end{align*}
and therefore it is straightforward to verify (see \eqref{value2}) that
\begin{equation}
\label{Vrew2}
v(t,x)=\sup_{0 \leq \tau \leq T-t}  \media \bigg[ \int_{0}^{\tau}e^{-\int_{0}^{s}r(t+u)du}H(t+s,X_s^x) ds\bigg].
\end{equation}

Notice that \eqref{Vrew2} includes a deterministic discount rate which is not time homogeneous. Optimal stopping problems of this kind are relatively rare in the literature. They feature technical difficulties which are more conveniently handled by considering a discounted version of the problem. Hence we introduce
\begin{equation}
\label{Vrew}
w(t,x):=e^{-\int_0^t r(s)ds}v(t,x)=\sup_{0 \leq \tau \leq T-t}  \media \bigg[ \int_{0}^{\tau}e^{-\int_{0}^{t+s}r(u)du}H(t+s,X_s^x) ds\bigg].
\end{equation}
Since the problem for $w$ is equivalent to the one for $v$ and $V$, from now on we focus on the analysis of \eqref{Vrew}.

From \eqref{value2} it is clear that $V(t,x) \geq G(t,x)$, for all $(t,x)\in[0,T]\times\realip$, therefore
$w$ is non-negative. Moreover it is straightforward to check that $w(t,x)$ is finite for all $(t,x)\in[0,T]\times\realip$, thanks to well known properties of $X$ and to Assumption \ref{ass1}.

As usual in optimal stopping theory, we let
\begin{equation}
\mathcal{C}= \big\{ (t,x)\in [0,T]\times \realip : w(t,x)>0 \big\} \label{setC}
\end{equation}
and
\begin{equation}
\mathcal{S}= \big\{ (t,x)\in [0,T]\times \realip : w(t,x)=0 \big\} \label{setS}
\end{equation}
be respectively the so-called continuation and stopping regions. We also denote by $\partial\cC$ the boundary of the set $\cC$
and we introduce the first entry time of $(t,X)$ into $\mathcal{S}$, i.e.
\begin{equation}
\label{entry}
\tau_*(t,x):= \inf \left \{ s \in [0,T-t] : (t+s,X_s^x) \in \mathcal{S} \right \},\quad\text{for $(t,x)\in[0,T]\times\realip$}.
\end{equation}

Since $(t,x)\mapsto H(t,x)$ is continuous, it is not difficult to see that for any fixed stopping time $\widetilde{\tau}\ge 0$, setting $\tau:=\widetilde{\tau}\wedge(\T-t)$, the map
\begin{align*}
(t,x)\mapsto \media \bigg[ \int_{0}^{\tau}e^{-\int_{0}^{t+s}r(u)du}H(t+s,X^x_s) ds\bigg]
\end{align*}
is continuous as well. It follows that $w$ is lower semi-continuous and therefore $\cC$ is open and $\cS$ is closed. Moreover, the finiteness of $w$ and standard optimal stopping results (see \cite[Cor.~2.9, Sec.~2]{PS}) guarantee that \eqref{entry} is optimal for $w(t,x)$.

For future frequent use we introduce here a new probability measure $\pprob$ defined by its Radon-Nikodym derivative
\begin{equation}
\label{Ptilde}
Z_t=\frac{d\,\pprob}{d\,\prob}\Big|_{\mathcal{F}_t}=\exp \left\{\sigma B_t-\frac{\sigma^2}{2} t \right\}, \quad t \geq 0,
\end{equation}
and notice that
\begin{align}\label{XZ}
X_t^{x}=x\,Z_t\, e^{(\theta-\alpha)t},\qquad t\ge 0.
\end{align}
It is well known that $\prob$ and $\pprob$ are equivalent on $\cF_t$ for all $t\in[0,\T]$.
\begin{remark}\label{rem:K0}
In case $K=0$ in \eqref{eq:P}, problem \eqref{Vrew} reduces to a deterministic problem. Noticing that
\[
\mathds{1}_{\{s<\tau\}}Z_s=\mathds{1}_{\{s<\tau\}}\media\left[Z_\tau|\cF_s\right]=\media\left[\mathds{1}_{\{s<\tau\}}Z_\tau|\cF_s\right]\,,
\]
because $\{\tau>s\}$ is $\cF_s$-measurable, and thanks to Fubini's theorem, one has
\begin{align*}
w(t,x)=&\,x \, \sup_{0 \leq \tau \leq T-t}  \media \bigg[ \int_{0}^{T-t}\!e^{-\int_{0}^{t+s}r(u)du}g(t+s)Z_s\mathds{1}_{\{s<\tau\}}e^{(\theta-\alpha)s} ds\bigg]\\[+4pt]
=&\,x \, \sup_{0 \leq \tau \leq T-t}  \media \bigg[ \int_{0}^{T-t}\!e^{-\int_{0}^{t+s}r(u)du}g(t+s)\media\left[ \mathds{1}_{\{s<\tau\}}Z_\tau\big|\cF_s\right] e^{(\theta-\alpha)s} ds\bigg]\\[+4pt]
=&\,x \, \sup_{0 \leq \tau \leq T-t}  \media \bigg[ Z_\tau\! \int_{0}^{\tau}\!e^{-\int_{0}^{t+s}r(u)du}g(t+s)e^{(\theta-\alpha)s} ds\bigg].
\end{align*}
Then, using $Z_\tau$ to change measure (cf.~\eqref{Ptilde}) we obtain
\[
w(t,x)=x \, \sup_{0 \leq \tau \leq T-t}  \mmedia \bigg[ \int_{0}^{\tau}e^{-\int_{0}^{t+s}r(u)du}g(t+s)e^{(\theta-\alpha)s} ds\bigg].
\]
The latter is equivalent to the deterministic problem of maximising the function
\[
F(t+\cdot\,):=\int_{0}^{\,\cdot}e^{-\int_{0}^{t+s}r(u)du}g(t+s)e^{(\theta-\alpha)s} ds.
\]
As a result the optimal annuitization time only depends on $t$ (as it happens in \cite{MY}).
\end{remark}

\begin{remark}\label{rem:T}
If we allow $T=+\infty$ and assume that
\[
\media\left[\int_0^\infty \! \!e^{-\int_{0}^{t}r(u)du}|H(t,X_t)|dt\right]<+\infty,
\]
then our problem \eqref{Vrew} remains well posed. We notice that the finite horizon $T$ only features in \eqref{Vrew} as part of the definition of the admissible stopping times. Then it is intuitively clear that the major mathematical challenges related to the time-dependency of \eqref{Vrew} arise from the properties of the map $t\mapsto e^{-\int_{0}^{t}r(u)du}H(t,x)$. Since such properties remain unchanged under $T=+\infty$, it follows that the analysis presented below may be extended to the case $T=+\infty$ after minor tweaks.
\end{remark}

The next proposition starts to analyse the regularity of $w$, and it provides a probabilistic characterization for its gradient which will be crucial for our subsequent analysis of the boundary of $\cC$.
\begin{proposition}\label{prop:bounds}
The value function $w$ is convex in $x$ for each $t\in[0,T]$ and it is locally Lipschitz continuous on $[0,T]\times\reali_+$. Moreover for a.e.\ $(t,x)\in[0,T]\times\realip$ we have
\begin{align}\label{wx}
w_x(t,x)=\mmedia \bigg[ \int_{0}^{\tau_*}e^{-\int_{0}^{t+s}r(u)du} g(t+s) e^{(\theta-\alpha)s}  ds\bigg]
\end{align}
and there exists a uniform constant $C>0$ such that
\begin{align}\label{bwt}
\hspace{-3pc}-C\left(1+\frac{1}{T-t}\right)\hspace{-3pt} \left(x\, \widetilde{\media}\,\tau_*+ \media\,\tau_*\right) \leq w_t(t,x) \leq C\hspace{-3pt} \left(x\, \widetilde{\media}\,\tau_*+ \media\,\tau_*\right)
\end{align}
\end{proposition}
\begin{proof}
\emph{1. (convexity)} Since $x\mapsto e^{-\int_0^t r(u)du}H(t,x)$ is linear then it is not difficult to show that for $x,y\in\realip$, $\lambda\in(0,1)$ and $x_\lambda:=\lambda x+(1-\lambda)y$ we have
\begin{align*}
&\media\left[\int^\tau_0e^{-\int_0^{t+s} r(u)du}H(t+s,X^{x_\lambda}_s)ds\right]\\
&=\lambda\media\left[\int^\tau_0e^{-\int_0^{t+s} r(u)du}H(t+s,X^{x}_s)ds\right]+(1-\lambda)\media\left[\int^\tau_0e^{-\int_0^{t+s} r(u)du}H(t+s,X^{y}_s)ds\right]\\
&\le \lambda w(t,x)+(1-\lambda)w(t,y),
\end{align*}
for any stopping time $\tau$. Taking the supremum over $\tau\in[0,T-t]$ the claim follows.
\vspace{+5pt}

\emph{2. (Lipschitz continuity)}
Fix $(t,x)\in[0,T]\times\realip$ and pick $\eps>0$. First we show that
\begin{align}\label{Lip-x}
|w(t,x\pm \eps)-w(t,x)|\le c\,\eps,
\end{align}
with $c>0$ independent of $(t,x)$.

Let $\tau_*=\tau_*(t,x)$ be optimal in $w(t,x)$, hence admissible and sub-optimal in $w(t,x+\eps)$, so that we have
\begin{align}\label{wx1}
w(t,x+\eps)-w(t,x)&\geq\media \bigg[ \int_{0}^{\tau_*}e^{-\int_{0}^{t+s}r(u)du}\left(H(t+s,X_s^{x+\eps}) -H(t+s,X_s^{x})\right) ds\bigg] \nonumber\\
&= \eps\media \bigg[ \int_{0}^{\tau_*}e^{-\int_{0}^{t+s}r(u)du} g(t+s) \frac{X_s^{x+\eps}-X_s^{x}}{\eps} ds\bigg]
\nonumber\\
&= \eps\media \bigg[ \int_{0}^{\tau_*}e^{-\int_{0}^{t+s}r(u)du} g(t+s) X^1_s ds\bigg]\\
&= \eps\mmedia \bigg[ \int_{0}^{\tau_*}e^{-\int_{0}^{t+s}r(u)du} g(t+s)e^{(\theta-\alpha)s}  ds\bigg],\nonumber
\end{align}
where for the last equality we used \eqref{XZ}. For the upper bound we repeat the above argument with $\tau_\eps^{+}:=\tau_*(t,x+\eps)$ optimal for $w(t,x+\eps)$ and find
\begin{align*}
w(t,x+\eps)-w(t,x)& \leq \eps \mmedia \bigg[ \int_{0}^{\tau_\eps^{+}}e^{-\int_{0}^{t+s}r(u)du} g(t+s) e^{(\theta-\alpha)s} ds\bigg].\nonumber
\end{align*}
Since $\tau_*$ and $\tau^+_\eps$ are smaller than $T-t$ we have $|w(t,x+\eps)-w(t,x)|\leq c\, \eps$ for a suitable $c>0$ independent of $(t,x)$.
By applying symmetric arguments we can also prove that $|w(t,x-\eps)-w(t,x)|\leq c\, \eps$ so that \eqref{Lip-x} holds.
For future reference we also notice that
\begin{align}\label{wx-}
w(t,x)-w(t,x-\eps)\leq&\media \bigg[ \int_{0}^{\tau_*}e^{-\int_{0}^{t+s}r(u)du}\left(H(t+s,X_s^{x}) -H(t+s,X_s^{x-\eps})\right) ds\bigg] \nonumber\\
=& \eps\mmedia \bigg[ \int_{0}^{\tau_*}e^{-\int_{0}^{t+s}r(u)du} g(t+s)e^{(\theta-\alpha)s}  ds\bigg]
\end{align}

Next we show that for all $\delta>0$, $x\in\realip$ and any $t\in[0,T-\delta]$ we have
\begin{align}\label{Lip-t}
|w(t\pm\eps,x)-w(t,x)|\le c_\delta\,\eps
\end{align}
for some $c_\delta>0$ only depending on $\delta$ and all $\eps\le T-t$.

Let $\tau_*=\tau_*(t,x)$ be optimal in $w(t,x)$, pick $\eps>0$ and define $\nu_\eps:=\tau_*\wedge (\T-t-\eps)$. Since $\nu_\eps$ is admissible and sub-optimal for $w(t+\eps,x)$ we get
\begin{align}
\label{uc0}
w&(t+\eps,x)-w(t,x)\nonumber\\
& \geq\media \bigg[ \int_{0}^{\nu_\eps}
e^{-\int_{0}^{t+\eps+s} r(u) du}H(t+\eps+s,X_s^{x})ds-\int_{0}^{\tau_*}e^{-\int_{0}^{t+s} r(u)du}
H(t+s,X_s^{x})ds\bigg] \nonumber\\
 &=\media \bigg[\int_{0}^{\nu_\eps} \Big( e^{-\int_{0}^{t+\eps+s} r(u) du}H(t+\eps+s,X_s^{x})-e^{-\int_{0}^{t+s} r(u)du}
H(t+s,X_s^{x}) \Big)ds  \bigg]\nonumber\\
 &\quad \hspace{0.2cm}-\media\bigg[\int_{\nu_\eps}^{\tau_*} e^{-\int_{0}^{t+s} r(u)du}
H(t+s,X_s^{x}) ds             \bigg].
\end{align}
Now we use that
\begin{align}
\label{uc4}
&\left|e^{-\int_{0}^{t+\eps+s} r(u) du}H(t+\eps+s,X_s^{x})-e^{-\int_{0}^{t+s} r(u)du}
H(t+s,X_s^{x})\right|\nonumber\\
&\leq \int^{\eps}_{0}\left|\frac{d}{d\,z}e^{-\int_{0}^{t+s+z} r(u)du}
H(t+s+z,X_s^{x}) \right|d\,z\nonumber\\
&\le\int^{\eps}_{0}\left| -r(t+s+z)H(t+s+z,X_s^{x}) +
\frac{\partial}{\partial z}H(t+s+z,X_s^{x})
\right|d\,z\nonumber\\
&\le c_1\left(1+ X_s^{x}\right)\eps
\end{align}
where the last estimate follows by \eqref{HG}, \eqref{g} and Assumption \ref{ass1}, with a uniform constant $c_1>0$ (recall that $r(\cdot)\ge 0$). Plugging the above expression in the first term of \eqref{uc0} and using the mean value theorem for the second term we get
\begin{align*}
w&(t+\eps,x)-w(t,x)\\
\ge& \media\bigg[- c_1\eps\int_0^{\nu_\eps} \left(1+ X_s^{x}\right)d\,s-H(t+\zeta,X_{\zeta}^{x})   (\tau_*-\nu_\eps)\bigg]\nonumber
\end{align*}
where $\zeta(\omega)\in(\nu_\eps(\omega),\tau_*(\omega))$. Notice that
\[
0\le (\tau_*-\nu_\eps)\le \eps\mathds{1}_{\{\tau_*\ge \T-t-\eps\}}
\]
and from \eqref{HG} one obtains the bound
\begin{align*}
\left| H(t+\zeta,X_{\zeta}^{x})   (\tau_*-\nu_\eps)\right|\leq c_1(1+ X_\zeta^{x})\eps \mathds{1}_{\{\tau_*\ge \T-t-\eps\}}.
\end{align*}
In conclusion, noticing that $\nu_\eps\le\tau_*$ and recalling \eqref{Ptilde} and \eqref{XZ}, by a change of measure we have
\begin{align}
\label{uc2}
w(t+\eps,x)-w(t,x)&\ge - c_1\eps\media\bigg[\int_0^{\tau_*} \left(1+ X_s^{x}\right)d\,s+(1+ X_\zeta^{x})\mathds{1}_{\{\tau_*\ge \T-t-\eps\}}\bigg]\nonumber\\
\ge& - C\eps\media\bigg[\int_0^{\tau_*} \left(1+ xZ_s\right)d\,s+(1+ xZ_\zeta)\mathds{1}_{\{\tau_*\ge \T-t-\eps\}}\bigg]\nonumber\\
\ge& - C\eps\bigg(
\media\,\tau_*+ x\,\mmedia\,\tau_*+\prob(\tau_*\ge \T-t-\eps)+ x\pprob(\tau_*\ge \T-t-\eps)
\bigg)
\end{align}
for a different constant $C>0$. Using the Markov inequality we obtain $\prob(\tau_*\ge \T-t-\eps)\le \media\,\tau_*/(\T-t-\eps)$ and $\pprob(\tau_*\ge \T-t-\eps)\le \mmedia\,\tau_*/(\T-t-\eps)$, which plugged back into \eqref{uc2} give
\begin{align}
\label{lbeps}
w(t+\eps,x)-w(t,x) \ge - C\eps\bigg(\media\,\tau_*+ x\,\mmedia\,\tau_*\bigg)\bigg(1+\frac{1}{\T-t-\eps}\bigg)
\end{align}
By using similar estimates and by observing that $\sigma_\eps^{+}:=\tau_*(t+\eps,x)$ is admissible and sub-optimal for $w(t,x)$ we get
\begin{align}
\label{wup}
w&(t+\eps,x)-w(t,x) \nonumber \\
& \leq\media \bigg[\int_{0}^{\sigma_\eps^{+}} \Big( e^{-\int_{0}^{t+\eps+s} r(u) du}H(t+\eps+s,X_s^{x})-e^{-\int_{0}^{t+s} r(u)du}
H(t+s,X_s^{x}) \Big)ds     \bigg]\nonumber\\
&\le c_1\eps\media\bigg[\int_0^{\sigma_\eps^{+}} \left(1+ X_s^{x}\right)d\,s\bigg]\le C\eps\bigg(\media\,\sigma_\eps^{+}+x\,\mmedia\,\sigma_\eps^{+}\bigg)
\end{align}
where we have used again \eqref{uc4} and the change of measure.

Symmetric arguments then give
\begin{align}
\label{wup1}
w(t,x)-w(t-\eps,x) \le C\eps\bigg(\media\,\tau_*+x\,\mmedia\,\tau_*\bigg)
\end{align}
and
\begin{align}
\label{lbld}
w(t,x)-w(t-\eps,x) \ge - C\eps\bigg(\media\,\sigma_\eps^{-}+ x\,\mmedia\,\sigma_\eps^{-}\bigg)\bigg(1+\frac{1}{\T-t}\bigg)
\end{align}
with $\sigma_\eps^{-}:=\tau_*(t,x-\eps)$.

Equations \eqref{lbeps}--\eqref{lbld} imply \eqref{Lip-t} and combining \eqref{Lip-x} and \eqref{Lip-t} we conclude that $w\in C([0,T]\times\realip)$, locally Lipschitz and differentiable a.e.\ in $[0,T)\times\realip$.
\vspace{+5pt}

\emph{3. (gradient bounds)} Let $(t,x)\in[0,T)\times\realip$ be a point of differentiability of $w$. Dividing \eqref{wx1} and \eqref{wx-} by $\eps$ and letting $\eps\to 0$ gives \eqref{wx}, as claimed. Moreover, dividing \eqref{lbeps} by $\eps$ and letting $\eps\to0$ we obtain the lower bound in \eqref{bwt}. Finally, dividing \eqref{wup1} by $\eps$ and letting $\eps\to 0$ we obtain the upper bound in \eqref{bwt}.

\end{proof}

Continuity of $w$ and standard optimal stopping theory guarantee that, for all $t\in[0,T]$, the process
\begin{align}\label{martw}
W_s:=w(t+s,X^x_{s})+\int_0^{s}e^{-\int_0^{t+u}r(v)dv}H(t+u,X^x_u)du,
\end{align}
is a continuous super-martingale for all $s\in[0,T-t]$ and $(W_{s\wedge\tau_*})_{s\ge0}$ is a martingale.

The next corollary follows by standard PDE arguments used normally in optimal stopping literature, see, e.g.\ \cite[Thm.\ 2.7.7]{KS98}.
\begin{corollary}\label{cor:C12}
The function $w$ is $C^{1,2}$ inside $\cC$ and it solves the following boundary value problem:
\begin{align}
\label{pde1} & (w_t+\cL w)(t,x) = -e^{-\int_0^tr(u)du} H(t,x), & (t,x)\in\cC, \\
\label{pde2} & w(t,x)= 0, & (t,x)\in\partial\cC\cap\{t<T\},\\
\label{pde3} & w(T,x) = 0, & x\in\realip.
\end{align}
\end{corollary}
\noindent It may appear that \eqref{pde1} is given in a slightly unusual form, but one should remember that $w(t,x)=e^{-\int_0^tr(u)du}v(t,x)$ (see \eqref{Vrew}) so that for $v$ we obtain the more canonical expression
\begin{align}\label{pde4}
(v_t+\cL v-r(\cdot))(t,x) = - H(t,x), \qquad (t,x)\in\cC.
\end{align}

The next technical lemma states some properties of $w$ that will be useful to study the regularity of the boundary $\partial \cC$.
Its proof is given in Appendix.
\begin{lemma}\label{lem:limw0}
Assume $g(t)< 0$ for $t\in(0,\T)$. Then
\begin{itemize}
\item[(i)] $x\mapsto w(t,x)$ is non-increasing for all $t\in[0,\T]$,
\item[(ii)] for any $t\in[0,\T]$ it holds
\begin{align}\label{limw0}
\lim_{x\to\infty}w(t,x)=0,
\end{align}
\item[(iii)] for all $t_1<t_2$ in $[0,T]$ we have $\cS\cap((t_1,t_2)\times\realip)\neq\emptyset$.
\end{itemize}
\end{lemma}
\noindent It is worth noticing that $(iii)$ does not exclude that there may exists $t\in(0,T)$ such that $\cS\cap(\{t\}\times\realip)=\emptyset$.


\section{Properties of the optimal boundary}\label{sec:freeb}
In this section we provide sufficient conditions for the boundary $\partial\cC$ to be represented by a function of time $b$. We establish connectedness of the sets $\cC$ and $\cS$ with respect to the $x$ variable and finally study Lipschitz continuity of $t\mapsto b(t)$. It is worth emphasizing that this study is mathematically challenging because of the lack of monotonicity of the map $t\mapsto b(t)$ and falls outside the scope of the existing probabilistic literature on optimal stopping and free boundary problems. In Section \ref{numerics} we show that the Gompertz-Makeham mortality law (a mainstream model in actuarial science) leads naturally to the set of assumptions that we make below.

An initial insight on the shape of $\cC$ is obtained by noticing that the set
\begin{align}\label{R}
\mathcal{R}:=\{(t,x)\in[0,T]\times\realip: H(t,x)>0\}
\end{align}
is contained in $\cC$. In fact if $(t,x)\in\mathcal R$ then
\[
w(t,x)\ge \media\left[\int_0^{\tau_{\mathcal{R}}}e^{-\int_0^{t+s}r(t+u)du}H(t+s,X^x_s)ds\right]>0
\]
with $\tau_{\mathcal{R}}$ the first exit time of $(t+s,X^x_s)$ from $\mathcal R$. For all $t\in[0,T]$ such that $g(t)\neq 0$ the boundary $\partial\mathcal{R}$ is given by the curve
\begin{align}\label{gamma}
\gamma(t):=-K \ell(t)/g(t).
\end{align}
Moreover, for each $t\in[0,T]$, we denote the $t$-section of $\mathcal R$ by
\begin{align}\label{Rt}
\mathcal{R}_t:=\{x\in\realip: (t,x)\in\mathcal R\}.
\end{align}

\begin{remark}\label{Ran}
\

{\bf (1)} Note that if $g(\cdot)$ and $K\ell(\cdot)$ in \eqref{HG} have the same sign on $[0,T]$ (i.e.~$\gamma(\,\cdot\,)\le 0$), then either $\mathcal{R}=\emptyset$ or $\mathcal{R}=[0,T]\times \realip$. In the former case $\mathcal{S}=[0,T]\times \realip$, while in the latter we have $\mathcal{C}=[0,T]\times \realip$.

{\bf (2)} Recall \eqref{g}. Consider $t\in[0,T)$ such that $g(t)<0$, i.e.\ $\ell(t)-\beta(t)-(\theta-\alpha)f(t)>0$. Then if $\theta\ge \alpha$ and $K>0$, i.e. $K$ is an acquisition fee, we have that $\mathcal R_t=[0,\gamma(t))$ and
\[
\gamma(t)=K \frac{\ell(t)}{\ell(t)-\beta(t)-(\theta-\alpha)f(t)}>K.
\]
The latter means that the annuity should not be purchased if the individual's wealth is less or equal than $K$ and the funds value has a positive trend (net of dividend payments).
\end{remark}
\noindent Motivated by {\bf (1)} in the remark above we will later assume that $\gamma(\cdot)>0$ on $[0,T]$.
\vspace{+4pt}

Now we illustrate sufficient conditions under which the probabilistic representation \eqref{wx} easily provides the shape of the continuation and stopping regions.
\begin{proposition}
\label{prop1}
If $g(\cdot)$ does not change its sign on $[0,T]$, the stopping region is characterized by a free boundary
\begin{align}
b:[0,T]\to\realip\cup\{+\infty\}.
\end{align}
In particular
\begin{itemize}
\item[(i)] If $g(t)\ge 0$ for all $t\in[0,T]$ then $\mathcal{S}=\{(t,x)\in [0,T]\times \realip: x\le b(t)\}$.
\item[(ii)] If $g(t)\le 0$ for all $t\in[0,T]$ then $\mathcal{S}=\{(t,x)\in [0,T]\times \realip: x\ge b(t)\}$.
\end{itemize}
\end{proposition}
\begin{proof}
In case \textit{(i)} $w_x(t,x)\geq 0$ and this implies
\begin{align*}
(t,x)\in\cC\implies (t,x')\in\cC\quad\text{for all $x'\ge x$},
\end{align*}
and the claim follows.
A similar argument applies to case \textit{(ii)}.
\end{proof}

For each $t\in[0,T]$ we denote $\mathcal S_t:=\{x\in\realip\,:\, (t,x)\in\cS\}$ and $\cC_t=\realip\setminus\cS_t$. These are the so-called $t$-sections of the stopping and continuation set, respectively. Clearly $\cS_t=[0,b(t)]$ under $(i)$ and $\cS_t=[b(t),+\infty)$ under $(ii)$ of the above proposition.

In the rest of the paper we make the next standing assumption. This will hold in all the results below without explicit mention. In Sections \ref{numerics} and \ref{sec:final} we discuss its range of applicability and some extensions.
\begin{Assumptions}\label{ass:g}
Assume $\gamma(t)>0$, for all $t\in[0,\T]$ and its limit $\gamma(T):=\lim_{t\to T}\gamma(t)$ exists (possibly infinite). Moreover we have
\[
\text{either {\bf (i)}: $g(t) > 0$ for all $t\in(0,\T)$, or {\bf (ii)}: $g(t) < 0$ for all $t\in(0,\T)$.}
\]
\end{Assumptions}

\begin{remark}\label{rem:gg}
It is worth discussing a financial/demographic interpretation of our assumption. We start by noticing that the money's worth $f(\cdot)$ should change slowly as function of time. Indeed in models with deterministic mortality force like ours it is unreasonable to imagine that the individual or the insurer may have drastic changes in their views on survival probabilities (this could instead be the case in stochastic mortality models).  Moreover, it is reasonable that an individual who is healthier than the average of the population at outset, will remain so over time, so that the map $t\mapsto f(t)-1$  should not change sign (this again would be less reasonable in stochastic mortality models).
Finally, we observe that as $t\to\infty$ one has $f(t)\to 1$, because eventually the individual and the insurer should agree on the survival probabilities.

As explained, $f'(\cdot)$ is small, hence the leading terms in $g(\cdot)$ are those involving $\beta$ and $\theta-\rho$ (see \eqref{g}). In particular they highlight the interplay between the financial and demographic risk for the individual.
On the one hand, $\beta(t)=\alpha+\mu^S(\eta+t)$ in \eqref{value2b} can be interpreted as the dividends arising from the financial investment adjusted by the demographic risk -- notice that an individual with high subjective mortality force is less likely to annuitize and would rather enjoy the return on a financial investment; on the other hand, $\theta-\rho$ is closely related to the risk premium of the financial investment. In $g(\cdot)$ these two terms are weighted exactly with $1-f(\cdot)$ and $f(\cdot)$, respectively. Moreover, $g(\cdot)$ in \eqref{Vrew2} represents the marginal gain, per unit of investment in stock, arising from delaying the annuitization by one time unit.

In light of all considerations so far (recall that $f(\cdot)-1$ does not change sign and $f'$ is almost negligible), it is clear from the formula in \eqref{g} that if $f(t)>1$ ($f(t)<1$) for all $t\in(0,\T)$, i.e.~the individual finds the annuity underpriced (overpriced), and $\theta<\rho$ ($\theta>\rho$) then $g(\cdot)$ remains negative (positive) on $[0,T]$. This conclusion is financially clear because an underpriced annuity and the prospect of a negative risk premium would give the investor a negative marginal gain (per unit of stock) from delaying the annuity purchase (and viceversa).

Finally, if $f(t)>1$ ($f(t)<1$) for all $t\in(0,\T)$ and $\theta>\rho$ ($\theta<\rho$) then $g(\cdot)$ may change its sign. However, if for example $t\mapsto \beta(t)$ is monotonic (i.e.~the mortality force is monotonic), then the change in sign is unlikely to occur more than once, due to the slow variation in time of the money's worth $f$. Motivated by this observation, we conduct numerical experiments in Section \ref{numerics} using Gompertz-Makeham mortality force that suggest that, when a change in the sign of $g(\cdot)$ occurs, this is likely to be over a time period longer than 20 years. It is therefore mainly for mathematical interest that we allow $g(\cdot)$ to vanish at time $T$. Indeed, we will also show in Section \ref{sec:final} that this extra flexibility enables us to extend our results to some cases when $g(\cdot)$ changes its sign once on $(0,T)$.
\end{remark}

\begin{remark}\label{rem:g}
Positivity of $\gamma(t)$ in Assumption \ref{ass:g} rules out the cases when $\mathcal{S}=[0,T]\times \realip$ or $\mathcal{C}=[0,T]\times \realip$ (see point {\bf (1)} of Remark \ref{Ran}).

Conditions (i) and (ii) in Assumption \ref{ass:g} are indeed supported by numerical experiments illustrated in Figures \ref{fig1} and \ref{fig3} in Section \ref{numerics}, where $g(\cdot)$ does not change sign or it changes sign at most once.
\end{remark}

Notice that Assumptions \ref{ass1} and \ref{ass:g} imply $\gamma\in C(0,T)$.
It is not hard to verify that the following properties hold.
\begin{lemma}\label{lem:obs}
\
\begin{itemize}
\item[(a)] If $g(\cdot)>0$ on $(0,\T)$ then for each $t\in(0,\T)$ there exists $g_0(t)>0$ such that
\begin{align}\label{obs1}
\int_0^sg(t+u)du\ge g_0(t)\,s\qquad\text{for all $s\in[0,\T-t]$}.
\end{align}
\item[(b)]
If $g(\cdot)<0$ on $(0,\T)$ then for each $t\in(0,\T)$ there exists $g_0(t)>0$ such that
\begin{align}\label{obs2}
\int_0^sg(t+u)du\le -g_0(t)\,s\qquad\text{for all $s\in[0,\T-t]$}.
\end{align}
\end{itemize}
\end{lemma}
The next simple lemma will be useful in what follows.
\begin{lemma}\label{lem:wt0}
If $K\ell(t)\le 0$ for all $t\in[0,T]$ then $\lim_{x\to0}w(t,x)=0$ for all $t\in[0,T]$. Hence $[0,T]\times\{0\}\subseteq\cS$.
\end{lemma}
\begin{proof}
Recalling \eqref{HG} and using dominated convergence we obtain
\begin{align*}
0\le \lim_{x\to 0}w(t,x)\le \lim_{x\to0}\int_0^{T-t}|g(t+s)|\media\left[X^x_s\right]ds=0.
\end{align*}
\end{proof}

Next we show that the optimal boundary is locally Lipschitz continuous on $[0,T)$, hence also bounded on any compact.
Some of the ideas in the proof below are borrowed from Theorem 4.3 in \cite{DAS}. However we cannot directly apply results of \cite{DAS} since Condition (D) therein corresponds to require that there exists $c>0$ such that
\[
\left|\tfrac{\partial}{\partial t}\left(e^{-\int_0^t r(s)ds}H(t,x)\right)\right|\le c(1+|g(t)|)\quad\text{for all $(t,x)\in[0,T]\times\realip$.}
\]
The latter bound is clearly impossible in our setting since the left hand side in the equation above is linear in $x$.
\begin{theorem}
\label{th1}
The optimal boundary $b(\cdot)$ is locally Lipschitz continuous on $[0,T)$.
\end{theorem}
\begin{proof}
In this proof most of the arguments are symmetric when we consider case $(i)$ and case $(ii)$ of Assumption \ref{ass:g}.

We start by noticing that, in case $(i)$ of Assumption \ref{ass:g}, Lemma \ref{lem:wt0} holds (because $\gamma(\cdot)>0$) and also $w_x> 0$ inside $\cC$ due to \eqref{wx}. The free boundary $b$ is the zero-level set of $w$, then continuity of $w$ and Lemma \ref{lem:wt0} imply that, for each $t\in[0,T)$, we can find $\delta>0$ sufficiently small so that the equation $w(t,x)=\delta$ has a solution $x=b_\delta(t)$.
The $\delta$-level set of $w$ is locally given by a continuous function $b_\delta$. Moreover $b_\delta(t)>b(t)\ge0$.
The family $(b_\delta)_{\delta>0}$ decreases as $\delta\to0$ so that its limit $b_0$ exists, it is upper semi-continuous as decreasing limit of continuous functions, and $b_0(t)\ge b(t)$. Since $w(t,b_\delta(t))=\delta$ it is clear that taking limits as $\delta\to0$ we get $w(t,b_0(t))=0$ and therefore $b_0(t)\le b(t)$ so that we conclude
 
\begin{align}\label{point}
\lim_{\delta\to0}b_\delta(t)=b(t)
\end{align}
for all $t\in[0,\T)$. Similarly, in case $(ii)$ of Assumption \ref{ass:g} we have $w_x<0$. Then $b_\delta(t)<b(t)$ and \eqref{point} holds (if $b(t)=+\infty$, the limit is $+\infty$). Finiteness of $b_\delta$ is always guaranteed by Lemma \ref{lem:limw0}.

Since $(t,b_\delta(t))\in\cC$ for all $t\in[0,T)$, and $w_x(t,b_\delta(t))\neq 0$ under either $(i)$ or $(ii)$ of Assumption \ref{ass:g}, an application of the implicit function theorem gives (recall that $w\in C^1$ in $\cC$)
\begin{equation}
\label{bprimo}
b^{\prime}_{\delta}(t)=-\frac{w_t(t,b_{\delta}(t))}{w_x(t,b_{\delta}(t))}, \quad t\in [0,\T).
\end{equation}
Next we will obtain a bound on $b'_\delta$, independent of $\delta$, for any bounded interval $[t_0,t_1]\subset(0,\T)$. This allows the use of Ascoli-Arzela's theorem to extract a sequence $(b_{\delta_j})_{j\ge 1}$ such that, as $j\to\infty$, $b_{\delta_j}$ converges uniformly on $[t_0,t_1]$  to a Lipschitz continuous function. Uniqueness of the limit and \eqref{point} imply that $b(\cdot)$ is locally Lipschitz.

To find the uniform bound on $b'_\delta$ we divide the proof in steps.
\vs{+6pt}

\emph{Step 1}. Let us start by observing that $w_x(t,b_{\delta}(t))>0$ in case $(i)$ of Assumption \ref{ass:g}, so that
\begin{align}\label{mod}
|b'_\delta(t)|=\frac{|w_t(t,b_\delta(t))|}{w_x(t,b_\delta(t))}.
\end{align}
To simplify notation, in what follows we set $x_\delta:=b_\delta(t)$ and $t$ is fixed. To estimate the numerator in \eqref{mod} we use the bound in \eqref{bwt},
thus obtaining
\begin{align}\label{modwt}
|w_t(t,x_\delta)|\leq C\left(1+\frac{1}{\T-t}\right)\!\left ( x_\delta\, \mmedia\, \tau_\delta + \media\, \tau_\delta\right),
\end{align}
where $\tau_\delta=\tau_*(t,x_\delta)$, for simplicity.

Plugging \eqref{modwt} inside \eqref{mod} we obtain
\begin{align}\label{mod2bis}
|b'_\delta(t)|\le C\left(1+\frac{1}{\T-t}\right)\frac{\left ( x_\delta\, \mmedia\, \tau_\delta + \media\, \tau_\delta\right)}{w_x(t,x_\delta)}.
\end{align}
A similar estimate but with the denominator replaced by $-w_x(t,x_\delta)$ can be obtained, up to obvious changes, in the setting of $(ii)$ in Assumption \ref{ass:g}, i.e.
\begin{align}\label{mod2}
|b'_\delta(t)|\le C\left(1+\frac{1}{\T-t}\right)\frac{\left ( x_\delta\, \mmedia\, \tau_\delta + \media\, \tau_\delta\right)}{-w_x(t,x_\delta)}.
\end{align}
\vs{+6pt}

At this point, in order to make \eqref{mod2bis} and \eqref{mod2} uniform in $\delta$, we must look at their limits as $\delta\to0$. For this we need to consider separately case $(i)$ and case $(ii)$ of Assumption \ref{ass:g}.
\vs{+6pt}

\emph{Step 2 - Case $(i)$}. Here $\cS_t=[0,b(t)]$ and, since $b\le\gamma$ and $\gamma\in C(0,T)$, then $b$ is locally bounded. Notice that, under $\pprob$ defined in \eqref{Ptilde}, the process $\widetilde{B}_s=B_s-\sigma s$, $s\ge 0$ is a Brownian motion, and the individual's wealth \eqref{wealth} started from $x_\delta$ may be written as
\[
dX_s^{x_\delta} = (\theta-\alpha+\sigma^2)X_s^{x_\delta} ds+\sigma X_s^{x_\delta} d\widetilde{B}_s.
\]
Then, for any $s\in[0,T-t]$ we have
\begin{align}\label{prob}
\pprob \left(\tau_\delta>s\right) =&\pprob \left(\inf_{0\leq u \leq s} \left(X^{x_\delta}_u-b(t+u)\right)>0\right)\nonumber\\
\geq& \prob  \left(\inf_{0\leq u \leq s} \left(X^{x_\delta}_u-b(t+u)\right)>0\right)=\prob \left(\tau_\delta>s\right)
\end{align}
and thus $\mmedia\, \tau_\delta\geq \media \tau_\delta$. Substituting this estimate in the numerator of \eqref{mod2bis} allows us to write
\begin{align}\label{mod3}
|b'_\delta(t)|\le C\left(1+\frac{1}{\T-t}\right)\!\left(1+x_\delta \right)\frac{\mmedia\, \tau_\delta}{w_x(t,x_\delta)}.
\end{align}
Recalling Assumption \ref{ass1}, \eqref{wx} and \eqref{obs1} it is now clear that there exists a constant $c>0$ independent of $(t,x_\delta)$ such that
\begin{align}\label{lip1}
w_x(t,x_\delta)\ge c\,\mmedia\left[\int_0^{\tau_\delta}g(t+s)ds\right]\ge c\,g_0(t) \mmedia\,\tau_\delta.
\end{align}
Hence plugging the latter into \eqref{mod3} we conclude that
\begin{align}\label{mod3bis}
|b'_\delta(t)|\le \frac{C}{c\,g_0(t)}\left(1+\frac{1}{\T-t}\right)\!\left(1+x_\delta \right)\le \frac{c'}{g_0(t)}\left(1+\frac{1}{\T-t}\right)\!\left(2+\gamma(t) \right)
\end{align}
with $c'>0$ a suitable constant and where we used $b_\delta(t) \leq 1+ \gamma(t)$ as $\delta\to 0$. Then the uniform bound in \eqref{mod3bis} implies that $b$ is locally Lipschitz as claimed.
\vs{+6pt}

\emph{Step 2 - Case $(ii)$}. Here $\cS_t=[b(t),+\infty)$. The analysis in this part is more involved than in the previous case. An argument similar to the one in \eqref{prob} gives $\media\,\tau_\delta\ge\mmedia\,\tau_\delta$ which unfortunately does not help with the estimate in \eqref{mod2}. So we need to proceed in a different way.

From \eqref{mod2} we see that, in order to bound $|b'_\delta(t)|$, we must bound $\mmedia\,\tau_\delta/|w_x(t,x_\delta)|$ and $\media\,\tau_\delta/|w_x(t,x_\delta)|$ (recall that $w_x\le 0$). The former can be bounded easily by using \eqref{wx} and \eqref{obs2}, since
\begin{align}\label{lip2}
w_x(t,x_\delta)\le c\,\mmedia\left[\int_0^{\tau_\delta}g(t+s)ds\right]\le -c\,g_0(t) \mmedia\,\tau_\delta
\end{align}
with suitable $c>0$. Then \eqref{lip2} implies
\begin{align}\label{lip3}
\frac{\mmedia\,\tau_\delta}{|w_x(t,x_\delta)|}\le \frac{1}{c\,g_0(t)}.
\end{align}

The other term requires more work because the expectation in the numerator is taken with respect to $\prob$ whereas the one in the denominator uses $\pprob$. Although we can still use \eqref{lip2} to estimate the ratio $\media\,\tau_\delta/|w_x(t,x_\delta)|$ now we end up with
\begin{align}\label{lip4}
\frac{\media\,\tau_\delta}{|w_x(t,x_\delta)|}\le \frac{\media\,\tau_\delta}{c\,g_0(t)\mmedia\,\tau_\delta}.
\end{align}
Our next task is to find a bound for the ratio $\media\,\tau_\delta\,/\,\mmedia\,\tau_\delta$.

Due to Assumption \ref{ass:g} there exists $a>0$ such that $\gamma(t)\ge 2a$ for $t\in[0,\T]$ and we denote
\begin{align*}
\tau_a=\inf\{t\ge0\,:\,X_t\le a\}.
\end{align*}
For this part of the proof it is convenient to think of $\Omega$ as the canonical space of continuous paths $\omega=\{\omega(t),t\ge0\}$ and denote $\vartheta_s$ the shifting operator $\vartheta_s \,\omega =\{\omega(s+t),t\ge 0\}$. Recall that $\mmedia_{t,x}[\,\cdot\,]=\mmedia[\,\cdot\,|X_t=x]$ and $\media_{t,x}[\,\cdot\,]=\media[\,\cdot\,|X_t=x]$. With this notation we must be careful that for fixed $(t,x_\delta)$ and any $s\ge 0$ we have $\prob(\tau_\delta>s)=\prob_{t,x_\delta}(\tau_*-t>s)$ and $\pprob(\tau_\delta>s)=\pprob_{t,x_\delta}(\tau_*-t>s)$, because
\[
\tau_*=\inf\{u\ge t\,:\,(u,X^{t,x_\delta}_u)\in\cS\}\quad\text{under $\prob_{t,x_\delta}$ and $\pprob_{t,x_\delta}$.}
\]

Our first estimate gives
\begin{align}\label{00}
\mmedia\,\tau_\delta=\mmedia_{t,x_\delta}(\tau_*-t)=\frac{1}{x_\delta}\media_{t,x_\delta}\left[e^{(\alpha-\theta)(\tau_*-t)}X_{\tau_*}({\tau_*}-t)\right]\ge\frac{c_1}{x_\delta}
\media_{t,x_\delta}\left[X_{\tau_*}\,({\tau_*}-t)\right]
\end{align}
for a suitable uniform constant $c_1>0$. Next we obtain
\begin{align}\label{01}
\media_{t,x_\delta}\left[X_{\tau_*}\,({\tau_*}-t)\right]=&\,\media_{t,x_\delta}\left[X_{\tau_*}\,({\tau_*}-t)\left(\mathds{1}_{\{{\tau_*}\le\tau_a\}}+
\mathds{1}_{\{{\tau_*}>\tau_a\}}\right)\right]\nonumber\\[+4pt]
\ge&\,a\media_{t,x_\delta}\left[({\tau_*}-t)\mathds{1}_{\{{\tau_*}\le\tau_a\}}\right]+\media_{t,x_\delta}\left[X_{\tau_*}\,({\tau_*}-t)\mathds{1}_{\{{\tau_*}>\tau_a\}}\right].	
\end{align}
The last term can be further simplified:
\begin{align}\label{02}
&\media_{t,x_\delta}\left[X_{\tau_*}\,({\tau_*}-t)\mathds{1}_{\{{\tau_*}>\tau_a\}}\right]\nonumber\\[+4pt]
&=\,\media_{t,x_\delta}\left[X_{\tau_*}\,({\tau_*}-t)\mathds{1}_{\{{\tau_*}>\tau_a\}}\left(\mathds{1}_{\{{\tau_*}<\T\}}+\mathds{1}_{\{{\tau_*}=\T\}}\right)\right]\nonumber\\[+4pt]
&\ge\,a\,\media_{t,x_\delta}\left[({\tau_*}-t)\,\mathds{1}_{\{{\tau_*}>\tau_a\}}\mathds{1}_{\{{\tau_*}<\T\}}\right]+
\media_{t,x_\delta}\left[X_{\tau_*}\,({\tau_*}-t)\,\mathds{1}_{\{{\tau_*}>\tau_a\}}\mathds{1}_{\{{\tau_*}=\T\}}\right]
\end{align}
where we have used that $\{\tau_*<T\}\subseteq\{X_{\tau_*}\ge \gamma({\tau_*})\}$ under $\prob_{t,x_\delta}$ and $\gamma({\tau_*})\ge 2a$. The last term in the above expression may be estimated by using iterated conditioning and the strong Markov property:
\begin{align}\label{03}
&\media_{t,x_\delta}\left[X_{\tau_*}\,({\tau_*}-t)\,\mathds{1}_{\{{\tau_*}>\tau_a\}}\mathds{1}_{\{{\tau_*}=\T\}}\right]\nonumber\\[+4pt]
&=\,\media_{t,x_\delta}\left[(\T-t)\,\mathds{1}_{\{{\tau_*}>\tau_a\}}\media_{t,x_\delta}\left(X_{\tau_*}\,\mathds{1}_{\{{\tau_*}=\T\}}\big|\cF_{\tau_a}\right)\right]\nonumber\\[+4pt]
&=\,\media_{t,x_\delta}\left[(\T-t)\,\mathds{1}_{\{{\tau_*}>\tau_a\}}\media_{t,x_\delta}\left(X_{\tau_a+{\tau_*}\circ\vartheta_{\tau_a}}\,\mathds{1}_{\{{\tau_*}\circ\vartheta_{\tau_a}
=\T-\tau_a\}}\big|\cF_{\tau_a}\right)\right]\\[+4pt]
&=\,\media_{t,x_\delta}\left[(\T-t)\,\mathds{1}_{\{{\tau_*}>\tau_a\}}\media_{\,\tau_a,X_{\tau_a}}\left(X_{{\tau_*}}\,\mathds{1}_{\{{\tau_*}=\T\}}\right)\right]\nonumber\\[+4pt]
&\ge\,\media_{t,x_\delta}\left[(\T-t)\,\mathds{1}_{\{{\tau_*}>\tau_a\}}\mathds{1}_{\{{\tau_*}=\T\}}
\media_{\,\tau_a,a}\left(X_{{\tau_*}}\,\mathds{1}_{\{{\tau_*}=\T\}}\right)\right]\nonumber\\[+4pt]
&\ge\,c_2\,\media_{t,x_\delta}\left[(\T-t)\,\mathds{1}_{\{{\tau_*}>\tau_a\}}\mathds{1}_{\{{\tau_*}=\T\}}\right],\nonumber
\end{align}
where
\begin{align}\label{pos}
c_2:=\inf_{0\le s\le \T}\media_{\,s,a}\left(X_{\T}\,\mathds{1}_{\{{\tau_*}=\T\}}\right)>0.
\end{align}
Notice that the strict positivity in \eqref{pos} may be verified by using the known joint law of the Brownian motion and its running supremum. Indeed, recalling that $\gamma(t)\ge 2a$ for all $t\in[0,T]$ we have
\begin{align*}
\media_{\,s,a}\left(X_{\T}\,\mathds{1}_{\{{\tau_*}=\T\}}\right)\ge \media_{s,a}\left(X_T\,\mathds{1}_{\{\sup_{s\le u\le T}X_u\le 2a\}}\right).
\end{align*}

From \eqref{01}, \eqref{02} and \eqref{03} we get
\begin{align}
\media_{t,x_\delta}\left[X_{\tau_*}\,({\tau_*}-t)\right]\ge&\, a\media_{t,x_\delta}\left[({\tau_*}-t)\mathds{1}_{\{{\tau_*}\le\tau_a\}}\right]+a\,\media_{t,x_\delta}\left[({\tau_*}-t)\,\mathds{1}_{\{{\tau_*}>\tau_a\}}
\mathds{1}_{\{{\tau_*}<\T\}}\right]\nonumber\\[+4pt]
&+c_2\,\media_{t,x_\delta}\left[(\T-t)\,\mathds{1}_{\{{\tau_*}>\tau_a\}}\mathds{1}_{\{{\tau_*}=\T\}}\right]
\ge c_3\,\media_{t,x_\delta}\,({\tau_*}-t)=c_3\,\media\,\tau_\delta
\end{align}
with $c_3=a\wedge c_2$. Now we plug the latter into \eqref{00} and then back in \eqref{lip4} and obtain
\begin{align}\label{05}
\frac{\media\,\tau_\delta}{|w_x(t,x_\delta)|}\le \frac{x_\delta}{c'\,g_0(t)}
\end{align}
with $c'=c\cdot c_1\cdot c_3>0$ a uniform constant. The above expression and \eqref{lip3} may now be substituted into \eqref{mod2} and give
\begin{align}\label{04}
|b'_\delta(t)|\le \frac{c''}{g_0(t)}\left(1+\frac{1}{\T-t}\right)b_\delta(t)
\end{align}
where $c''>0$ is a suitable constant.

We recall that, in case $(ii)$ of Assumption \ref{ass:g}, Lemma \ref{lem:limw0}-$(iii)$ implies
that for any $0\le t_0<t_1<T$ we can find $t_2\in(t_0,t_1)$ such that $b_\delta(t_2)\le b(t_2)<+\infty$. The latter and \eqref{04} allow to apply Gronwall's inequality to obtain that, for any $0\le t_0<t_1<T$, there exists a constant $c_{t_0,t_1}>0$, independent of $\delta$ and such that
\begin{align}\label{unifb}
\sup_{t_0\le t\le t_1}\big|b_\delta(t)\big|\le c_{t_0,t_1}\,.
\end{align}
The same bound therefore holds for the boundary $b$, and it shows that $b$ is bounded on any compact. Moreover \eqref{04} and \eqref{unifb} also give a uniform bound for $b'_\delta$ on $[t_0,t_1]$ as needed.
\end{proof}

Lipschitz continuity of the boundary has important consequences regarding the regularity of the value function $w$, which we summarize below.
\begin{proposition}\label{thm:C1}
The value function $w$ is continuously differentiable on $[0,\T)\times\realip$. Moreover $w_{xx}$ is continuous on the closure of the set $\cC\cap\{t<\T-\eps\}$ for all $\eps>0$.
\end{proposition}
\begin{proof}
Corollary \eqref{cor:C12} tells us that $w_t$ and $w_x$ are continuous at all points in the interior of $\cC$ and of $\cS$, and thus it remains to analyze the regularity of $w$ across the boundary. In order to do that it is crucial to observe that since $t\mapsto b(t)$ is locally Lipschitz, the law of iterated logarithm implies that $\tau_*$ is indeed equal to the first time $X$ goes strictly below the boundary, in case $(i)$, or strictly above the boundary, in case $(ii)$. In other words the first entry time to $\cS$ is equal to the first entry time to its interior part. This is an important fact that can be used to prove that $(t,x)\mapsto\tau_*(t,x)$ is continuous $\prob$-a.s., and it is zero at all boundary points~(see for example \cite[Lemma 5.1 and Proposition 5.2]{DeAE16}).

Fix $(t_0,x_0)\in\partial\cC\cap\{t<T\}$ and take a sequence $(t_n,x_n)_{n\ge1}\!\subset\!\cC$ with $(t_n,x_n)\to (t_0,x_0)$ as $n\to \infty$. Continuity of $(t,x)\mapsto \tau_*(t,x)$ and the discussion above imply that $\tau_*(t_n,x_n)\to 0$, $\prob$-a.s.\ as $n\to \infty$. The latter and formulae \eqref{wx} and \eqref{bwt} give $w_x(t_n,x_n)\to 0$ and $w_t(t_n,x_n)\to 0$. Since $(t_0,x_0)$ and the sequence $(t_n,x_n)$ were arbitrary we get $w\in C^{1}([0,T)\times\realip)$.

The final claim regarding continuity of $w_{xx}$ follows from \eqref{pde1}. We know from Corollary \ref{cor:C12} that $w_{xx}$ is continuous in $\cC$. Moreover for any $(t_0,x_0)\in\partial\cC\cap\{t<\T\}$ we can take limits in \eqref{pde1} as $(t,x)\to(t_0,x_0)$ with $(t,x)\in\cC$ and use that $w(t_0,x_0)=w_x(t_0,x_0)=w_t(t_0,x_0)=0$ to obtain
\begin{align}\label{wxxb}
\lim_{\cC\owns(t,x)\to (t_0,x_0)}\frac{\sigma^2x^2}{2}w_{xx}(t,x)=\frac{\sigma^2x_0^2}{2}w_{xx}(t_0,x_0)=-e^{-\int_0^{t_0}r(u)du}H(t_0,x_0)
\end{align}
which proves our claim.
\end{proof}

As it will be discussed in Section \ref{numerics}, for the numerical evaluation of the optimal boundaries it is important to find the limit value of $b(t)$ when $t\to\T$. This will be analyzed in the next proposition. For future use we introduce the function (recall $g_0$ as in Lemma \ref{lem:obs})
\begin{align}\label{u0}
u_0(t):=\frac{1+\gamma(t)}{g_0(t)},\qquad t\in[0,T].
\end{align}
\begin{proposition}\label{prop:limb1}
Recall that $\gamma(T):=\lim_{t\to T}\gamma(t)$ exists but could be infinite (Assumption \ref{ass:g}). If $u_0\in L^1(0,\T)$ then $b(\cdot)$ has limit
\begin{align}\label{limb2}
\lim_{t\uparrow\T}b(t)=\gamma(\T).
\end{align}

\end{proposition}
\begin{proof}
Let us first consider case $(i)$ in Assumption \ref{ass:g}.
Here we recall the notation $x_\delta=b_\delta(t)$, $\tau_\delta=\tau_*(t,x_\delta)$ used in the proof of Theorem \ref{th1}.

From \eqref{bprimo} and the upper bound in \eqref{bwt} we get (recall that $w_x>0$)
\begin{align*}
b'_\delta(t)\ge& -\frac{C\left( x_\delta \,\mmedia\, \tau_\delta + \media\, \tau_\delta\right)}{w_x(t,x_\delta)}
\ge- C\left(1+ x_\delta\right) \,\frac{\mmedia\, \tau_\delta}{w_x(t,x_\delta)}
\end{align*}
where in the last inequality we have used $\mmedia\,\tau_\delta\ge\media\,\tau_\delta$ which follows from \eqref{prob}. Employing \eqref{lip1} and letting $\delta \rightarrow 0$ we find
\begin{align}
b'(t)\ge - \frac{C}{c\,g_0(t)}\left(1+ b(t)\right)\ge -\frac{C'}{g_0(t)}\left(1+ \gamma(t)\right), \qquad\text{for a.e.~$t\in[0,\T)$,}
\end{align}
where we have also used that $b$ is bounded from above by $\gamma$ (see \eqref{gamma}) and $C'=C/c$. Recalling $u_0(\cdot)$ and setting $\hat{b}(t)=b(t)+C'\,\int_0^t u_0(s) ds$, the mapping $t\rightarrow \hat{b}(t)$ is increasing. Thus, $\lim_{t\uparrow \T} \hat{b}(t) $ exists and since $u_0$ is integrable and positive on $[0,\T]$ then the limit $\lim_{t\uparrow \T} b(t)$ exists as well.

Notice that $b(t)\le \gamma(t)$ for all $t\in[0,\T)$ and therefore $b(\T-)\le \gamma(\T)$. Recall that $\gamma(\T)>0$ due to Assumption \ref{ass:g}. For the proof of \eqref{limb2} we follow the approach of \cite{DeA15}. Arguing by contradiction we assume $b(\T-)<\gamma(\T)$. Then we can pick $a,b$ such that $b(\T-)<a<b<\gamma(\T)$ and $t^{\prime}<\T$ such that
$(a,b)\times [t^{\prime},T)\subset \mathcal{C}$. It is convenient here to work with $v$ rather than $w$ (see \eqref{Vrew}) and to refer to \eqref{pde4}.

Let $\varphi \in C^{\infty}(a,b)$ with $\varphi\geq 0$ and define
$F_{\varphi}(s):=\int_{a}^{b}v_t(s,y)\varphi(y)dy$. Now, denoting by $\mathcal{L}^*$ the adjoint of the operator $\mathcal{L}$, we have
\begin{align*}
\lim_{s\uparrow T}F_{\varphi}(s)=&\lim_{s\uparrow T}\int_{a}^{b}\left[-H(s,y)-\mathcal{L}v(s,y)+r(s)v(s,y)\right]\varphi(y)dy\\
=&\lim_{s\uparrow T}\int_{a}^{b}\left[-H(s,y)\varphi(y)+v(s,y)\left( -\mathcal{L}^*+ r(s)\right)\varphi(y)\right]dy\\
=&\int_{a}^{b}\left[-H(T,y)\varphi(y)\right]dy>0
\end{align*}
since $v(T,y)=0$ and $H(T,y)<0$ for $y\in(a,b)$. From the above we also deduce that $F_{\varphi}(\cdot)$ is continuous up to $T$, thus there exists $\delta>0$ such that $F_{\varphi}(s)>0$, for
$s\in[T-\delta,T]$, and
\begin{align*}
0<\int_{\T-\delta}^{\T} F_{\varphi}(s) ds=\int_{a}^{b} \varphi(y) \left[v(\T,y)-v(\T-\delta,y)\right]dy=-\int_{a}^{b} \varphi(y) v(\T-\delta,y)dy<0
\end{align*}
since $v(\T-\delta,y)>0, \ y\in(a,b)$. Hence a contradiction.

Consider now case $(ii)$ of Assumption \ref{ass:g}.
Since $b(t)\ge \gamma(t)$ for all $t\in[0,\T]$, it is obvious that \eqref{limb2} holds if $\gamma(T)=+\infty$. For the case of $\gamma(T)<+\infty$ instead we recall \eqref{bprimo} and notice that \eqref{bwt} implies
\begin{align}
b'_\delta(t)\le C\left(\,\frac{b_\delta(t)\mmedia \tau_\delta+\media\tau_\delta}{|w_x(t,b_\delta(t))|}\,\right),\qquad t\in[0,\T).
\end{align}
If we now recall \eqref{lip3} and \eqref{05} we find
\begin{align}\label{upbdb}
b_\delta'(t)\le c\,\frac{b_\delta(t)}{g_0(t)}\qquad \text{for}\ t\in[0,\T)
\end{align}
for a suitable $c>0$.

Since $u_0\in L^1(0,T)$ and $\gamma\in C((0,T])$ then $1/g_0(t)$ is integrable on $[0,\T]$. Gronwall's lemma applied to \eqref{upbdb} guarantees that
\[
b_\delta(t)\le b_\delta(t_0)e^{\int_{t_0}^t\tfrac{c}{g_0(s)}ds}\qquad \text{for $t\in[0,T]$}
\]
with $t_0\in(0,T)$ arbitrarily chosen. It was shown in Theorem \ref{th1} that  $b$ is bounded on any compact subset of $(0,T)$ and therefore taking limits as $\delta\to0$ and recalling \eqref{point} we find that $b$ is indeed bounded on $[0,\T]$. This fact and \eqref{upbdb} in turn imply that $b'(t)\le c_0/g_0(t)$ a.e.\ on $[0,\T]$ for a suitable $c_0>0$ only depending on $t_0$.

We can therefore define $\hat{b}(t):=b(t)-c_0\int_0^t 1/g_0(s) ds$ and the mapping $t\mapsto\hat{b}(t)$ is non-increasing so that $\hat{b}(T-):=\lim_{t\to\T}\hat{b}(t)$ exists. Since $1/g_0(t)$ is integrable and positive on $[0,\T]$ then $b(t):=\lim_{t\to\T}b(t)$ exists too, it is finite and $b(\T-)\ge \gamma(\T)$. To prove \eqref{limb2} we argue by contradiction, assuming $b(\T-)>\gamma(\T)$. Then following analogous arguments to those employed in the proof of case $(i)$ above we reach the desired contradiction.
\end{proof}
\begin{remark}
In general the map $s\mapsto e^{-\int_{0}^{t+s}r(u)du}H(t+s,x)$ is not monotonic. As a consequence it becomes extremely challenging (if possible at all) to determine whether $b(\cdot)$ is monotonic or not. For numerical evidence of non-monotonic boundaries see Section \ref{numerics}.
\end{remark}

\subsection{Characterisation of the free boundary and of the value function}
In the next theorem we will find non-linear integral equations that characterise uniquely the free boundary and the value function. Here we notice that all regularity properties of $w$ obviously transfer to $V$ of \eqref{value2}, due to \eqref{w} and \eqref{Vrew}. In particular we notice that $V\in C^{1}([0,T)\times\realip)$ and $V_{xx}\in L^\infty_{loc}([0,T)\times\realip)$, with the only discontinuity of $V_{xx}$ occurring across $\partial\cC$. It is important to remark that Corollary \ref{cor:C12} and the remaining properties of $w$ studied above imply that indeed $V$ solves \eqref{HJB} in the almost everywhere sense (more precisely at all points $(t,x)\notin\partial\cC$).
\begin{theorem}
\label{ThVolt1}
For all $(t,x)\in[0,T]\times\realip$, the value function \eqref{value2} has the following representation
\begin{align}
\label{reprV1}
V(t,x)= \media &\, \bigg[ e^{-\int_{0}^{T-t}r(t+u)du}G(T,X^x_{T-t})\notag\\
&+\int_{0}^{T-t}e^{-\int_{0}^{s}r(t+u)du}\Big(\beta(t+s)X^x_s -H(t+s,X^x_s)\mathds{1}_{\{(t+s,X^x_s) \in \cS\}}\Big) ds\bigg].
\end{align}

Recall \eqref{u0} and assume $u_0\in L^1(0,T)$. In case $(i)$ of Assumption \ref{ass:g} (resp.\ in case $(ii)$), the optimal boundary $b$ is the unique continuous solution, smaller (resp.\ larger) than $\gamma$, of the following non-linear integral equation: for all $t\in[0,T]$
\begin{align}
\label{reprb}
G(t,b(t))\!=\! \media &\, \bigg[ e^{-\int_{0}^{T-t}r(t+u)du}G(T,X^{b(t)}_{T-t})\notag\\
& +\!\int_{0}^{T-t}\!e^{-\int_{0}^{s}r(t+u)du}\!\left(\beta(t+s)X^{b(t)}_s \!-\!H(t+s,X^{b(t)}_s)\mathds{1}_{\{(t+s,X^{b(t)}_s) \in \cS\}}\right) ds\bigg],
\end{align}
with $b(T-)=\gamma(T)$ (cf.~\eqref{limb2}).
\end{theorem}
\begin{proof}
Here we only show how to obtain \eqref{reprV1}. Then inserting $x=b(t)$ in \eqref{reprV1} and using that $V(t,b(t))=G(t,b(t))$, we get \eqref{reprb}. The proof of uniqueness is standard in modern optimal stopping literature and it dates back to \cite{Pe05} (for more examples see \cite{PS}). However, we provide the full argument in appendix for the interested reader.

Let $\left( V^{(n)}\right)_{n \geq0}$ be a sequence with $V^{(n)} \in C^{\infty}([0,T)\times \realip)$ such that (see \cite[Sec.~7.2]{GT})
\begin{align}\label{sob1}
\left( V^{(n)},V^{(n)}_x, V^{(n)}_t\right)\rightarrow \left( V,V_x, V_t\right)
\end{align}
as $n\uparrow \infty$, uniformly on any compact set, and
\begin{align}\label{sob2}
\lim_{n\rightarrow \infty}( V^{(n)}_{xx}-V_{xx})(t,x)=0\quad\text{for all $(t,x)\notin\partial\cC$}.
\end{align}

Let $\left( K_{m}\right)_{m \geq0}$ be an increasing sequence of compact sets converging to $[0,T]\times\realip$ and define
\[
\tau_m= \inf \left \{ s>0: (t+s,X_s^x) \notin K_{m} \right \}\wedge (T-t).
\]
Then an application of It\^o calculus gives
\begin{align}
\label{reprV2}
V^{(n)}(t,x)=& \media \bigg[ e^{-\int_{0}^{\tau_m}r(t+u)du}V^{(n)}(t+\tau_m,X^x_{\tau_m})\notag\\
 &  \hspace{0.5cm}-\int_{0}^{\tau_m}e^{-\int_{0}^{s}r(t+u)du}\left(V^{(n)}_t+\mathcal{L}V^{(n)}-r(t+s)V^{(n)} \right)(t+s,X^x_s)ds\bigg].
\end{align}

We want to let $n\uparrow \infty$ and use \eqref{sob1} and \eqref{sob2}, upon noticing that $(t+s,X_s)$ lies in a compact for $s\le \tau_m$, and that its law is absolutely continuous with respect to the Lebesgue measure on $[0,T]\times\realip$. The latter in particular implies $\prob\big((t+s,X^x_s)\in\partial\cC\big)=0$ for all $s\in[0,T-t)$ and enables the use of \eqref{sob2}. Recall that $V$, $V_x$, $V_t$ and $V_{xx}$ are locally bounded. Then, from dominated convergence and \eqref{reprV2} we obtain
\begin{align*}
V(t,x)=&\lim_{n\to\infty}V^{(n)}(t,x)\\
=&\media \bigg[ e^{-\int_{0}^{\tau_m}r(t+u)du}V(t+\tau_m,X^x_{\tau_m})\\
& \hspace{+15pt}-\int_{0}^{\tau_m}e^{-\int_{0}^{s}r(t+u)du}\left(V_t+\mathcal{L}V-r(t+s)V \right)(t+s,X^x_s)ds\bigg].
\end{align*}
Therefore, using \eqref{HJB} (or equivalently Corollary \ref{cor:C12}) we also find
\begin{align*}
V(t,x)\!=& \media\bigg[ e^{-\int_{0}^{\tau_m}r(t+u)du}V(t+\tau_m,X^x_{\tau_m})+\!\!\int_{0}^{\tau_m}\!\!e^{-\int_{0}^{s}r(t+u)du}\beta(t+s)X^x_s \mathds{1}_{\{(t+s,X^x_s) \in \cC\}}ds\\
 &\hspace{+15pt}+ \int_{0}^{\tau_m}e^{-\int_{0}^{s}r(t+u)du}\big[\beta(t+s)X^x_s-H(t+s,X^x_s)\big]\mathds{1}_{\{(t+s,X^x_s) \in \cS \}} ds\bigg].
\end{align*}

Finally we take $m\uparrow \infty$, use that $\tau_m \uparrow (T-t)$ and $V(T,x)=G(T,x)$, and apply dominated convergence to obtain \eqref{reprV1}.
\end{proof}

\section{Numerical findings}
\label{numerics}
Here we apply the results obtained in the previous sections to some situations of practical interest. For simplicity, throughout the section we take $\rho=\widehat \rho$.
A standard choice to model the force of mortality is the so-called Gompertz-Makeham law, which corresponds to
\begin{align}
\label{force}
\mu(t)=A+BC^t,
\end{align}
where $A$, $B$ and $C$ are real-valued and are estimated by statistical data of the population.
For simplicity, here we assume $A=0.00055845, \ B=0.000025670, \ C=1.1011$ as in \cite{HD}\footnote{These figures are those used by the Belgian regulator, Arr\^et\'e Vie 2003}. Time is measured in years and we consider two different scenarios:
\begin{description}
 \item[(a)] $f(t)\equiv f>0$ (see \eqref{f}) and $\mu^S(\cdot)=\mu(\cdot)$;
 \item[(b)] $\mu^O(\cdot)=\mu(\cdot)$ and $\mu^S(\cdot)=(1+\bar{\mu})\mu^O(\cdot)$ with $\bar{\mu} \in (-1,+\infty)$
 \end{description}

In the first scenario the money's worth function \eqref{f} is constant. If the individual believes she is healthier than the average in the population, then $\mu^S(\,\cdot\,) < \mu^O(\,\cdot\,)$ and therefore $f>1$. Conversely, for an individual who is pessimistic about her health $\mu^S(\,\cdot\,) > \mu^O(\,\cdot\,)$ and therefore $f<1$.  It is important to notice that the the function $g$ in \eqref{g} is \emph{monotonic} increasing (decreasing) if $f$ is a constant smaller (greater) than $1$.
\begin{figure}[ht!]
\centering
\includegraphics[scale=0.3]{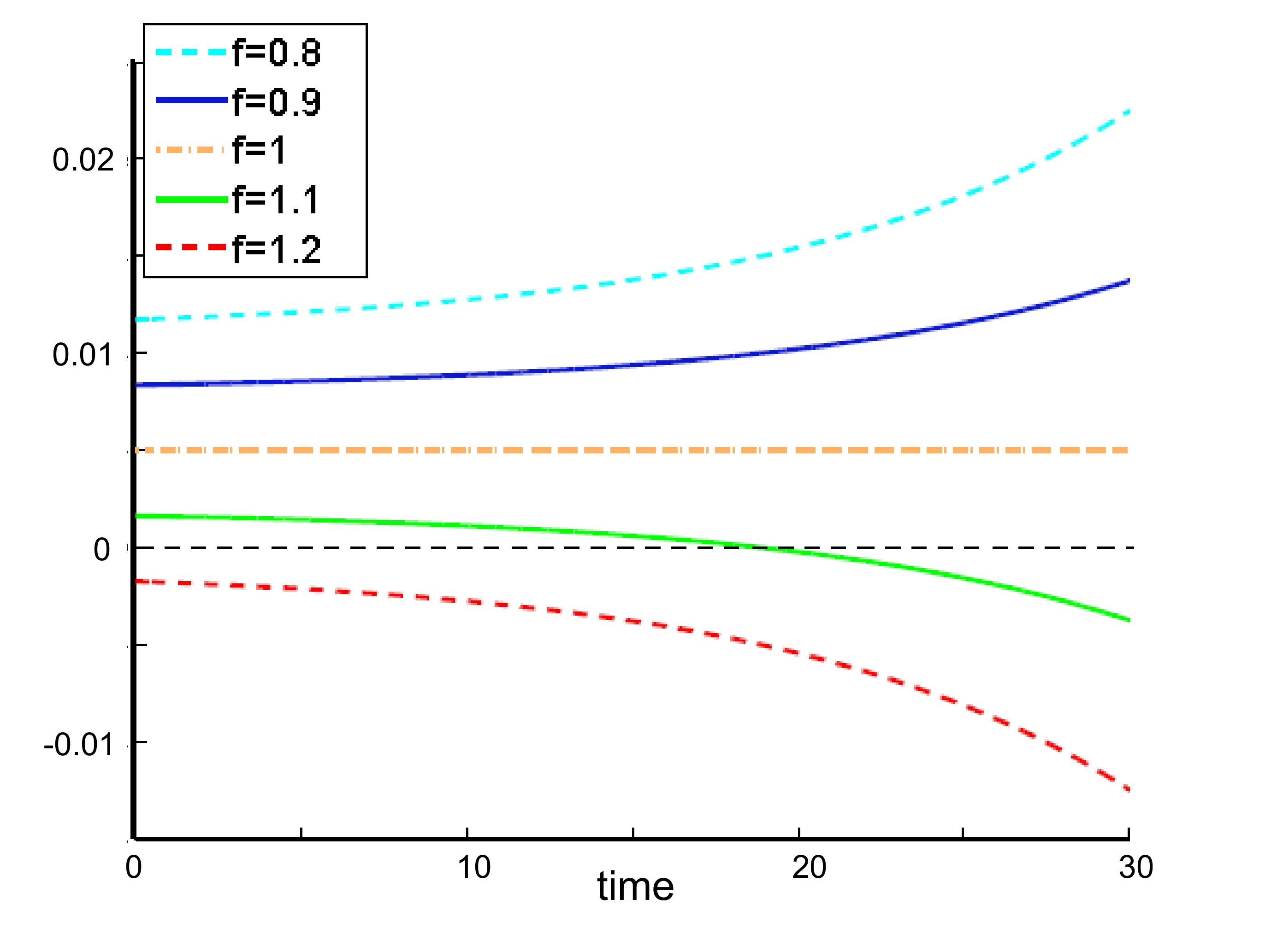}
\caption{Scenario (a). Behavior of function $g$.}
\label{fig1}
\end{figure}

The second scenario, uses the so-called \textit{proportional hazard rate transformation}, introduced in actuarial science by Wang (\cite{Wang}, see also \cite{MY}). If $\bar{\mu}<0$ (resp.\ $\bar \mu>0$), the individual considers herself healthier (resp.\ unhealthier) than the average.
The limit case  $\bar{\mu}\to -1$, is not relevant in practice as it corresponds to an individual whose life-expectancy is infinite. Similarly, the case $\bar{\mu}\rightarrow +\infty$ is also irrelevant in practice as it corresponds to an individual who believes she is about to die.
An important difference between scenarios (a) and (b) above is that, in the latter, the money's worth $f$ varies over time. In particular, if $\bar{\mu}<0$ (resp.\ $\bar \mu>0$) then
$f(t)>1$ (resp.\ $f(t)<1$), for all $t\in[0,T]$.

We notice that in all our numerical experiments the function $\ell(\,\cdot\,)$ of \eqref{g} is positive on $[0,T]$, so that the sign of $\gamma$ in \eqref{gamma} only depends on that of $K$ and $g$.
For the reader's convenience we also recall the standard numerical algorithm to compute \eqref{reprb} (under the assumptions of Theorem \ref{ThVolt1}). We take a equally-spaced partition $0=t_0<t_1<\ldots<t_{n-1}<t_n=T$ with $h:=t_{i+1}-t_i$. Starting from $b(T)=\gamma(T)$, for $i=1,2,\ldots n$ we solve
\begin{align*}
G(t_{i},b(t_i))= e^{-\int_{0}^{(n-i)h}r(t_i+u)du}\media\, \Big[ G(T,X^{b(t_i)}_{(n-i)h}) \Big]+h\sum^{n-i}_{k=1}Y(t_i,b(t_i),t_i+kh,b(t_i+kh))
\end{align*}
where
\begin{align*}
Y(t,b(t),\,&t+s, b(t+s))\\
&:=e^{-\int_{0}^{s}r(t+u)du}\Big(\beta(t+s)\media\big[X^{b(t)}_{s}\big]
-\media\Big[ H(t+s,X^{b(t)}_{s})\mathds{1}_{\{X^{b(t)}_{s} \le b(t+s)\}}\Big]\Big).
\end{align*}
Notice that the above formula is intended for $\cS_t=[0,b(t)]$. To deal with $\cS_t=[b(t),+\infty)$ we must change the indicator variable in the last expression in the obvious way.
\vspace{+5pt}

Unless otherwise specified, in what follows we take $T=30$, $\eta=50$, $\theta=4.5\%$, $\alpha=3.5\%$, $\sigma=10\%$ and $\widehat\rho=\rho=4\%$ (we take $\widehat\rho=\rho$ just for simplicity).
\vspace{+5pt}

\noindent\textbf{Scenario} (a).
\vspace{+5pt}

\begin{figure}[ht!]
\centering
\includegraphics[scale=0.5]{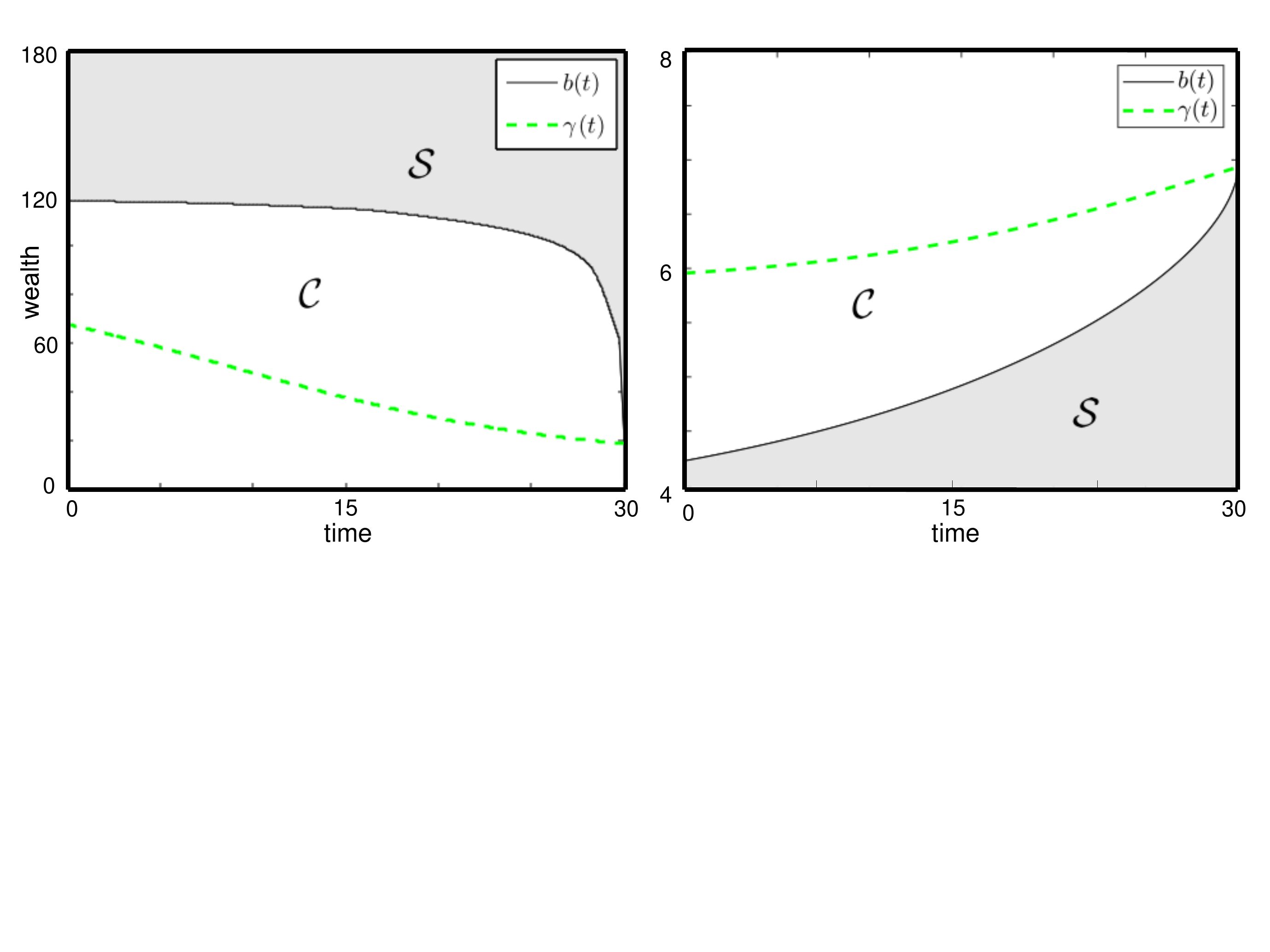}
\vspace{-4.5cm}
\caption{Scenario (a). The optimal annuitization regions and boundaries for $f=1.2, K=2$ (left plot) and for $f=0.8, K=-2$ (right plot).}
\label{fig2}
\end{figure}
In Figure \ref{fig1} the function $g$ in \eqref{g} is computed for different values of the constant $f$. As noticed in \cite{HD}, it is reasonable to expect that the value of $f$ is close to $1$. Notice that if $f$ is high enough ($f=1.2$) then $g$ remains always negative even if $\theta>\rho$.
We observe that $g$ varies slowly and in most cases it does not change sign, as required in Assumption \ref{ass:g}. However, if $f$ changes its sign (at most once, since $g$ is monotonic) we can still apply our methods as described in fuller details in Section \ref{sec:final}.

Figure \ref{fig2} shows the optimal annuitization regions and boundaries examining two of the cases considered in Figure \ref{fig1}, where $g$ is either negative ($f=1.2$) or positive ($f=0.8$) on $[0,T]$. We note that in the former case, if $K \leq 0$ then an immediate annuitization is optimal for all $(t,x)\in[0,T]\times \realip$ (see Remark \ref{Ran}). In the presence of a fixed acquisition fee $K>0$, instead, individuals annuitize as soon as the fund's value exceeds the boundary $b(\cdot)$ (left plot in Fig.~\ref{fig2}). On the other hand, in the case $f=0.8$, if $K \geq 0$ then the annuity is never purchased (Remark \ref{Ran}). Instead, in presence of a fixed tax incentive $K<0$, the annuity is purchased as soon as the fund's value falls below the boundary $b(\cdot)$ (right plot in Fig.~\ref{fig2}).
\vspace{+5pt}

\noindent\textbf{Scenario} (b).
\vspace{+5pt}

In Figure \ref{fig3} we look at scenario (b) and the function $g$ is plotted for different values of the constant $\bar{\mu}$ in cases $\theta<\rho$ (left plot) and $\theta>\rho$ (right plot). We note that, in a right neighbourhood of zero $g$ has the same sign of $\theta-\rho$. For most parameter choices, either $g$ does not change sign or it changes it once. 
Notice however that the change in sign occurs for $t\approx 22$ years, hence Assumption \ref{ass:g} is very reasonable.

\begin{figure}[ht!]
\centering
\includegraphics[scale=0.5]{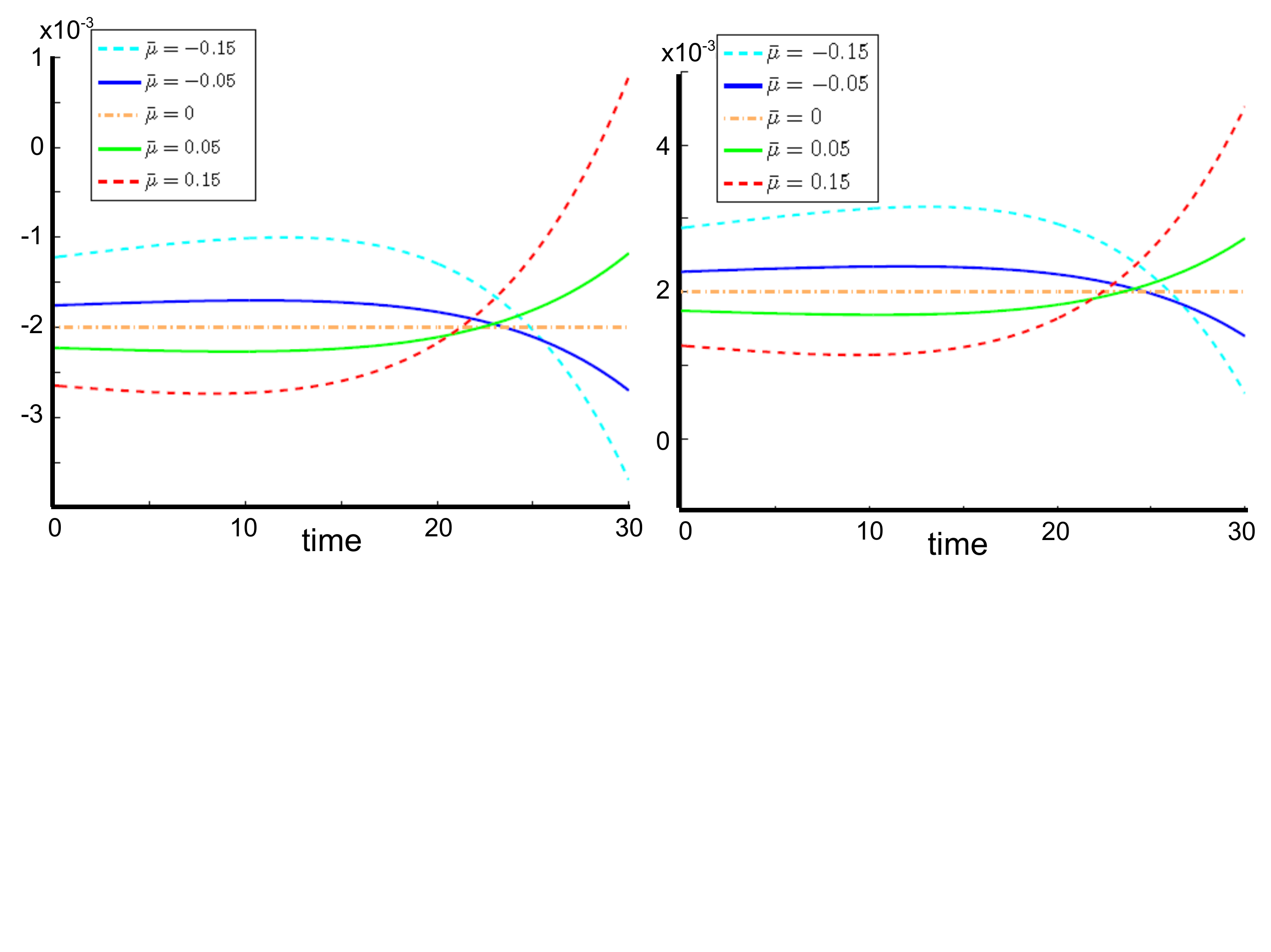}
\vspace{-4.5cm}
\caption{Scenario (b). Behaviour of function $g$ for $\theta<\rho$ (left plot) and for $\theta>\rho$ (right plot).}
\label{fig3}
\end{figure}

Optimal annuitization regions and their boundaries are presented in Figure \ref{fig4}.
\begin{figure}[ht!]
\centering
\includegraphics[scale=0.5]{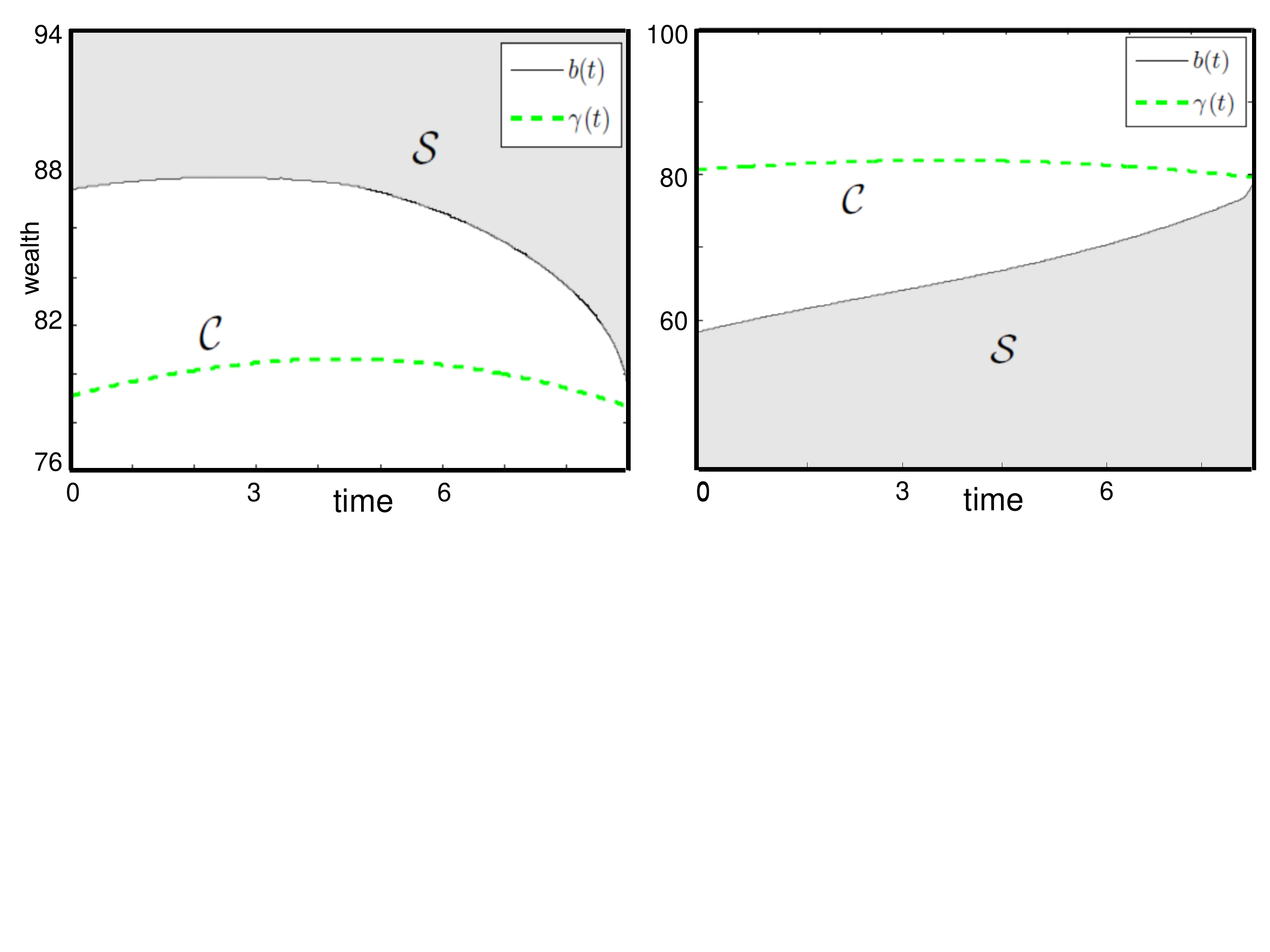}
\vspace{-4.8cm}
\caption{Scenario (b). The optimal annuitization regions and boundaries for $\bar{\mu}=-0.05, K=2, \theta<\rho$ (left plot) and for $\bar{\mu}=0.05, K=-2, \theta>\rho$ (right plot).}
\label{fig4}
\end{figure}
Re\-mar\-ka\-bly, we observe a non-monotonic optimal boundary in the left plot.
In the case $g<0$, an accurate numerical solution of the integral equation for $b$ needs a very fine partition of the interval $[0,T]$, which results in long computational times. We believe this is due to the steep gradient of $b$ near $T$ and to its lack of monotonicity. To simplify our analysis (which is intended for illustrative purpose only) we consider shorter time horizons for the investor than in scenario (a), i.e.~$T=9$ years.

\begin{figure}[ht!]
\centering
\includegraphics[scale=0.3]{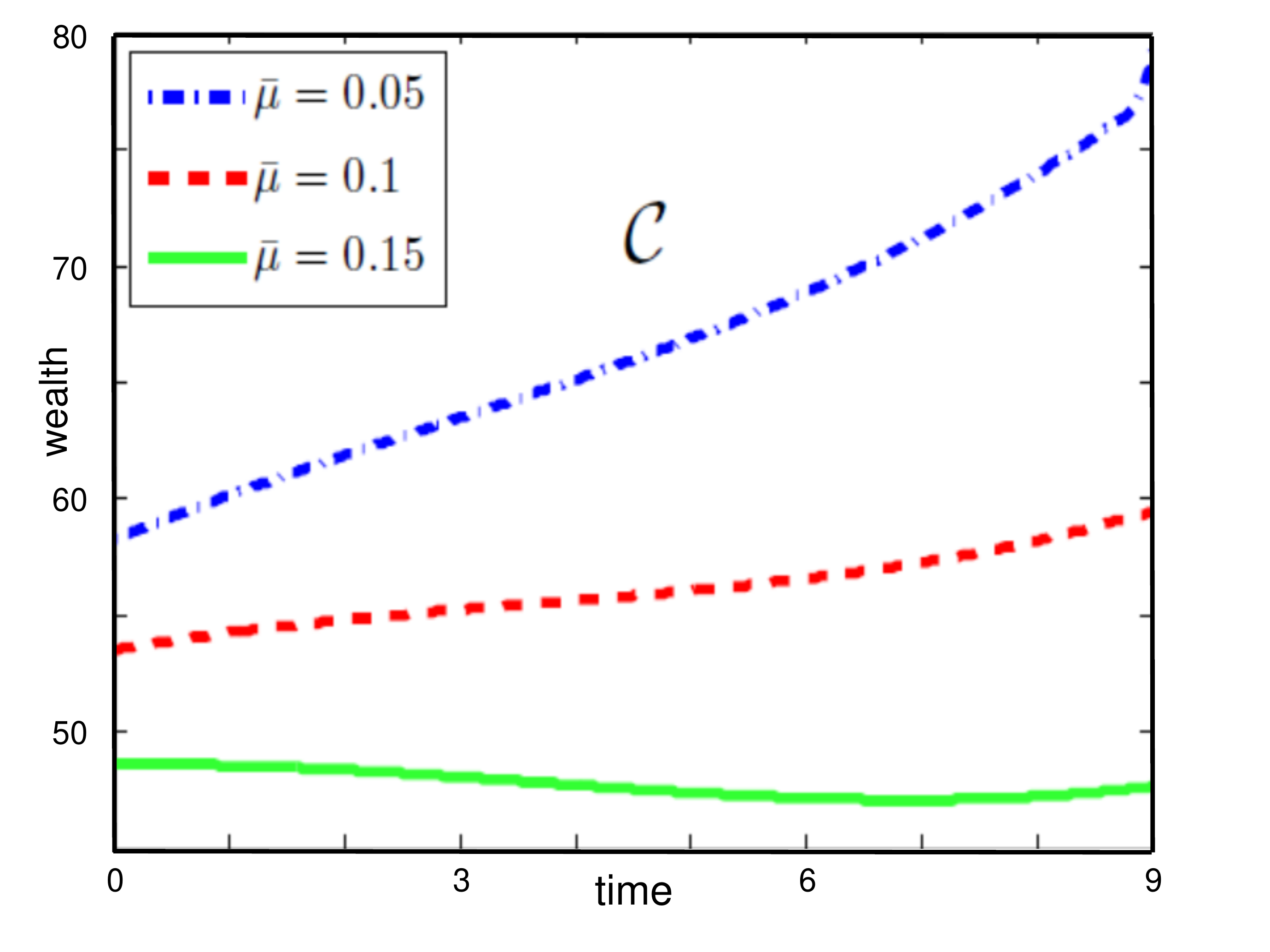}
\caption{Scenario (b). Sensitivity of the optimal boundary with respect to $\bar \mu$ for $\theta>\rho$ and $K=-2$.}
\label{fig5}
\end{figure}
In Figure \ref{fig5} we also study sensitivity of the annuitization boundary with respect to $\bar \mu>0$.
Recall that as $\bar \mu$ increases the individual considers herself increasingly unhealthier than the average population. As a results, we observe that the boundary $b(\cdot)$ is pushed downward and the continuation region expands. This is intuitively clear because annuities are financially less appealing for individuals with shorter (subjective) life expectancy.

\section{Final remarks and extensions}\label{sec:final}

As discussed in Remark \ref{rem:gg}, our main technical assumption (Assumption \ref{ass:g}) is supported by numerical experiments on the Gompertz-Makeham mortality law. The latter is widely used in the actuarial profession, hence it is a natural choice from the modelling point of view. We also notice that Assumption \ref{ass:g} allows for $g(T)=0$. This condition is mainly of mathematical interest. Indeed it enables extensions of our results to cover examples where $g$ is monotonic on $[0,T]$ and it changes its sign once. The latter examples are observed in Figures \ref{fig1} and \ref{fig3} (although the change of sign occurs only on rather long time horizons, e.g.\ $T>20$ years). On the other hand it appears that the function $\ell(\cdot)$ (see \eqref{g}) is positive in all of our numerical experiments.

Here we explain how our results can cover extensions to the case of $g(\cdot)$ changing its sign once. We shall consider separately the case of $K<0$ and $K>0$. From now on we assume that $g$ is monotonic and there exists $t_0\in(0,T)$ such that $g(t_0)=0$. We also assume that $\ell(t)> 0$ for $t\in[0,T]$ and recall $\mathcal R$ and $\gamma$ from \eqref{R} and \eqref{gamma}.
\vspace{+5pt}

\noindent\textbf{Case} $K<0$.
\vspace{+5pt}

1. (\emph{$g(\cdot)$ decreasing}).  In this setting we have $\gamma(t)>0$ for $t\in[0,t_0)$ with $\gamma(t)\uparrow +\infty$ as $t\uparrow t_0$. Moreover $\mathcal R$ lies above the curve $\gamma$ on $[0,t_0)$. For $t\in[t_0,T]$ we have $\mathcal R=\emptyset$ and therefore $\cS\cap\{t\ge t_0\}=[t_0,T]\times\realip$ (see Remark \ref{Ran}). This implies that $t_0$ is an effective time horizon for our optimization problem \eqref{value2} since it is optimal to immediately stop for any later time. From a mathematical point of view this means that we can equivalently study \eqref{value2} with $T$ replaced by $t_0$. On the effective time horizon $[0,t_0]$ part $(i)$ of Assumption \ref{ass:g} holds and we can repeat the analysis carried out in Sections \ref{sec:value} and \ref{sec:freeb}.
\vspace{+5pt}

2. (\emph{$g(\cdot)$ increasing}). In this setting we have $\mathcal R=\emptyset$ for $t\in[0,t_0]$, while in the interval $(t_0,T]$ we have $\gamma(\cdot)>0$ with $\gamma(t)\uparrow +\infty$ as $t\downarrow t_0$ and $\mathcal{R}$ lies above the curve $\gamma$. We can therefore study problem \eqref{value2} on the restricted time horizon $(t_0,T]$ where our Assumption \ref{ass:g} holds. This will give us $\cS_t=[0,b(t)]$ for $t\in(t_0,T]$ and all the results from the previous sections hold.

Moreover, we can show that $\cS_t=[0,b(t)]$ also for $t\in[0,t_0]$ (with $b(t)$ possibly infinite). For that we recall that $w(t,\,\cdot\,)$ is convex for each $t\in[0,T]$ (Proposition \ref{prop:bounds}) and $w(t,0+)=0$ (see Lemma \ref{lem:wt0}). The latter two properties imply $\cS_t=[0,b(t)]$ for $t\in[0,t_0]$ as claimed and $w_x(t,\,\cdot\,)\ge 0$ for all $t\in[0,t_0]$.

In summary, $\cS_t=[0,b(t)]$ for all $t\in[0,T]$ and most of the analysis in Section \ref{sec:freeb} carries over to this setting. However, it should be noted that methods used in Theorem \ref{th1} only allow to establish Lipschitz continuity of $b$ in $[t_0,T]$. A complete study of the boundary in $[0,t_0]$ requires new methods and we leave it for future work.
\vspace{+5pt}

\noindent\textbf{Case} $K>0$.
\vspace{+5pt}

1. (\emph{$g(\cdot)$ decreasing}). Here $\mathcal R\cap\{t\le t_0\}=[0,t_0]\times\realip$, so that $\cC\cap\{t\le t_0\}=[0,t_0]\times\realip$ and it is optimal to delay the annuity purchase at least until $t_0$, regardless of the dynamics of the fund's value. On $(t_0,T]$ instead we find that $\gamma(\cdot)>0$ with $\gamma(t)\uparrow +\infty$ as $t\downarrow t_0$ and $\mathcal{R}$ lies below the curve $\gamma$. From a mathematical point of view this means that we only need to study our problem \eqref{value2} on the restricted time horizon $(t_0,T]$ where our Assumption \ref{ass:g} holds.
\vspace{+5pt}

2. (\emph{$g(\cdot)$ increasing}). Here $\mathcal R\cap\{t\ge t_0\}=[t_0,T]\times\realip$, so that $\cC\cap\{t\ge t_0\}=[t_0,T]\times\realip$ and for $t\ge t_0$ it is optimal to delay the annuity purchase until maturity $T$, regardless of the dynamics of the fund's value. On $[0,t_0)$ instead we find that $\gamma(\cdot)>0$ with $\gamma(t)\uparrow +\infty$ as $t\uparrow t_0$ and $\mathcal{R}$ lies below the curve $\gamma$.

This case is much more challenging and we could not cover it with methods developed so far. We leave it for future research but nevertheless we would like to make an observation to highlight a key difficulty.

Take $t \in [0,t_0)$. The martingale property in \eqref{martw} allows us to rewrite problem \eqref{Vrew2} as follows
\begin{align}
\label{Vrew3}
v(t,x)=\sup_{0 \leq \tau \leq t_0-t}  \media \bigg[ &\int_{0}^{\tau}e^{-\int_{0}^{s}r(t+u)du}H(t+s,X_s^x) ds\\
&+e^{-\int_{0}^{t_0-t}r(t+u)du}v(t_0,X_{t_0-t}^x)\mathds{1}_{\{\tau = t_0-t\}} \bigg],\notag
\end{align}
where $v(t_0,x)$ may be explicitly calculated. In fact, from \eqref{g} we get
\begin{align*}
v(t_0,x)&=\media \hspace{-1pt} \bigg[\hspace{-1pt} \int_{0}^{T-t_0}e^{-\int_{0}^{s}r(t_0+u)du}\Big(K\ell(t_0+s)+g(t_0+s)X_s^x \Big)ds\bigg] \\
&= \int_{0}^{T-t_0}e^{-\int_{0}^{s}r(t_0+u)du}\Big(K\ell(t_0+s)+x\,g(t_0+s)e^{(\theta-\alpha)s} \Big)ds.
\end{align*}
So it is clear that $v(t_0,x)\!=\!c_1\!+\!c_2x$ with $c_1$ and $c_2$ positive constants that depend on $t_0$.

Due to the geometry of $\mathcal R$ we would expect that the stopping region lies somewhere above the curve $\gamma$ for $t\in[0,t_0)$. However if we now compute $v_x$ as in Proposition \ref{prop:bounds} it turns out that
\begin{align*}
v_x(t,x)=\mmedia \bigg[\hspace{-1pt} \int_{0}^{\tau_*}e^{-\int_{0}^{s}r(t+u)du}g(t+s)e^{(\theta-\alpha)s} ds\hspace{-1pt}+\hspace{-1pt}e^{-\int_{0}^{\tau_*}r(t+u)du+(\theta-\alpha)(t_0-t)}c_2\mathds{1}_{\{\tau_* = t_0-t\}} \hspace{-1pt} \bigg].
\end{align*}
Since $g(\cdot)<0$ on $[0,t_0)$ and $c_2>0$ it is no longer obvious that $v_x$ is negative on $[0,t_0)\times\realip$. This would have been a sufficient condition to guarantee that $\cS_t$ and $\cC_t$ are connected for all $t\in[0,t_0)$.

Noticing that $H(t,x)$ and $v(t_0,x)$ are linear in $x$, the asymptotic behaviour of $v(t,x)/x$ as $x\to\infty$ and convexity of $v(t,\cdot)$ (see Proposition \ref{prop:bounds}) suggest that, for a fixed $t\in[0,t_0)$, we should have $\cS_t=[b_1(t),b_2(t)]$ with $b_2(t)$ possibly infinite. This however leaves several open questions concerning the actual shape of $\cS$ and the regularity of its boundary. A complete answer to such questions requires the use of different methods and we leave it for future work.

\subsection{A comment on the regularity of the optimal boundary}

One of the main mathematical challenges in this work was the lack of monotonicity for the optimal boundary, which we managed to overcome by showing that the boundary is indeed locally Lipschitz.
While our methodology only relies on stochastic calculus, the idea of looking at the implicit function theorem and to provide bounds on $b'_\delta$ in \eqref{bprimo}, somehow comes from PDEs (we refer to \cite{DAS} for a more extensive review on the topic). In particular we were inspired by \cite{SS91}, where a variational inequality associated to an optimal stopping problem\footnote{A very minor difference with our work is that \cite{SS91} studies a minimisation problem.} is studied (see~eq.~(1.3) therein). Interestingly, our assumptions are rather weaker than those contained in \cite[pp.~376-377]{SS91}, so in this sense our probabilistic method extends results which were obtained in \cite{SS91} with purely analytical tools. The parallel with our notation is better understood if we use the problem formulation \eqref{Vrew}, although using \eqref{value2} is clearly equivalent.

First of all the problem studied in \cite{SS91} is for Brownian motion with drift and both the gain function and running cost in the optimal stopping problem are required to be of polynomial growth. This is immediately violated in our case where $(X_t)_{t\ge0}$ is the exponential of a Brownian motion with drift and the function $H$ in \eqref{Vrew} is linear in $X$. More importantly, what really is crucial for the proof of Lipschitz regularity of the free boundary in \cite{SS91} is that some particular lower bounds on the gradient of the running cost should hold. According to \cite{SS91}, in our Assumption \ref{ass:g} we would additionally need
\begin{align*}
e^{-\int_0^tr(u)du}|g(t)|\ge c\left[1+e^{-\int_0^tr(u)du}\left|H_t(t,x)-r(t)H(t,x)\right|\right]
\end{align*}
for all $(t,x)\in[0,T]\times\realip$ and for some $c>0$. Since the left-hand side of the above expression is independent of $x$, it is clear that the bound cannot hold.


\section*{Appendix}

\subsection*{Admissible stopping times}

Here we use an argument from \cite{CG12} to confirm that we incur no loss of generality in optimising over $\cT_{t,T}$ instead of using stopping times with respect to $(\cG_t)_{t\ge 0}$ where $\cG_t=\cF_t\vee\sigma(\{\Gamma_D>s\}\,,\,0\le s\le t)$. The key is that for any $(\cG_t)_{t\ge0}$-stopping time $\tau'$ there is a $(\cF_t)_{t\ge 0}$-stopping time $\tau$ such that $\tau\wedge\Gamma_D=\tau'\wedge\Gamma_D$, $\prob\times\Q^S$-a.s.~(see \cite[Ch.~VI.3, p.~370]{Pr}).

Then, letting $\cT'_{0,T}$ be the set of $(\cG_t)_{t\ge0}$-stopping times in $[0,T]$ we have

\begin{align*}
&\sup_{\tau' \in \cT'_{0,T} }\cE^S\bigg[\! \int_{0}^{\Gamma_D \wedge \tau' }\!\!e^{-\rho s}\alpha X_s ds+\mathds{1}_{\{\Gamma_D \leq \tau' \}}e^{-\rho \Gamma_D} X_{\Gamma_D} +P_{\eta+\tau'}\hspace{-3pt}\int_{\Gamma_D \wedge \tau'}^{\Gamma_D}e^{-\rho s}ds\bigg]\nonumber\\
&=\sup_{\tau' \in \cT'_{0,T} }\cE^S\bigg[\! \int_{0}^{\Gamma_D \wedge \tau' }\!\!e^{-\rho s}\alpha X_s ds+\mathds{1}_{\{\Gamma_D = \tau'\wedge\Gamma_D \}}e^{-\rho \Gamma_D} X_{\Gamma_D} +P_{\eta+\tau'\wedge\Gamma_D}\hspace{-3pt}\int_{ \tau' \wedge\Gamma_D}^{\Gamma_D}e^{-\rho s}ds\bigg]\nonumber\\
&=\sup_{\tau \in \cT_{0,T} }\cE^S\bigg[\! \int_{0}^{\Gamma_D \wedge \tau }\!\!e^{-\rho s}\alpha X_s ds+\mathds{1}_{\{\Gamma_D = \tau\wedge\Gamma_D \}}e^{-\rho \Gamma_D} X_{\Gamma_D} +P_{\eta+\tau\wedge\Gamma_D}\hspace{-3pt}\int_{ \tau \wedge\Gamma_D}^{\Gamma_D}e^{-\rho s}ds\bigg]\nonumber\\
&=\sup_{\tau \in \cT_{0,T} }\cE^S\bigg[\! \int_{0}^{\Gamma_D \wedge \tau }\!\!e^{-\rho s}\alpha X_s ds+\mathds{1}_{\{\Gamma_D \leq \tau \}}e^{-\rho \Gamma_D} X_{\Gamma_D} +P_{\eta+\tau}\hspace{-3pt}\int_{\Gamma_D \wedge \tau}^{\Gamma_D}e^{-\rho s}ds\bigg],
\end{align*}
where in the first equality we used that $\{\Gamma_D\le \tau'\}=\{\Gamma_D= \tau'\wedge\Gamma_D\}$ and that
\[
P_{\eta+\tau'}\hspace{-3pt}\int_{ \tau' \wedge\Gamma_D}^{\Gamma_D}e^{-\rho s}ds=\mathds{1}_{\{\tau'\le \Gamma_D\}}P_{\eta+\tau'}\hspace{-3pt}\int_{ \tau' \wedge\Gamma_D}^{\Gamma_D}e^{-\rho s}ds.
\]
The same argument may be repeated for stopping times in $\cT'_{t,T}$ upon taking expectation conditioned to $\cF_t\cap\{\Gamma_D>t\}$.

\subsection*{Proof of Lemma \ref{lem:limw0}}
Monotonicity of $w(t,\,\cdot\,)$ follows immediately from \eqref{wx}.

We now prove $(iii)$. For that we need a preliminary result and we introduce
\[
\gamma_{0}(t):=\inf\{x\in\realip\,:\,H(t,x)<-\delta_0\}
\]
for a fixed $\delta_0>0$. Since $g(t)<0$ for $t\in(0,T)$ it easy to verify that
\begin{align}\label{gammadelta}
\gamma_{0}(t)=-(\delta_0+K\ell(t))/g(t)
\end{align}
and $H(t,x)<-\delta_0$ for $x\in(\gamma_{0}(t),+\infty)$.

We define the stopping time
\begin{align}\label{tau0}
\tau_0=\tau_0(t,x):=\inf\{s\ge 0\,:\,X^x_s\le \gamma_{0}(t+s)\}\wedge (\T-t)
\end{align}
and notice that $\tau_0(t,x)\to\T-t$ in probability as $x\to\infty$. In fact for any $\eps>0$ there exists $\gamma_\eps\in(0,+\infty)$ such that $\gamma_{0}(t+s)\le \gamma_\eps$ for all $s\le \T-t-\eps$, therefore
\begin{align}\label{limprob}
\prob(\T-t-\tau_0\ge \eps)=&\,\prob(\tau_0\le \T-t-\eps)\le\, \prob\left(\inf_{0\le s\le \T-t-\eps}X^x_s\le \gamma_\eps\right)\\[+4pt]
\le&\, \prob\left(\inf_{0\le s\le \T-t-\eps}X^1_s\le \frac{\gamma_\eps}{x}\right)\nonumber
\end{align}
and clearly the last term goes to zero as $x\to\infty$.
\vs{+4pt}

Notice now that if there exists $t_*<T$ such that $[t_*,\T)\times\realip\subset\cC$, then for all $s\ge t_*$ one has
\begin{align}\label{form1}
w(s,\gamma_0(s))=&\,\media\left[\int_0^{\T-s}e^{-\int_0^{s+u}r(v)dv}H(s+u,\gamma_0(s)X^{1}_u)du\right]\\[+4pt]
=&\,\gamma_0(s)\!\int_0^{\T-s}\!\!e^{-\int_0^{s+u}r(v)dv}g(s\!+\!u)e^{(\theta-\alpha)u}du\!+\!K\!\int_0^{\T-s}\!\!e^{-\int_0^{s+u}r(v)dv}\ell(s\!+\!u)du.\nonumber
\end{align}
Noticing that
\begin{align}\label{form2}
\gamma_0(s)\int_0^{\T-s}g(t+u)du=-(\delta_0+K\ell(s))\frac{1}{g(s)}g(s+\xi_s)(\T-s)
\end{align}
for some $\xi_s\in[0,\T-s]$, we immediately conclude that
\begin{align}\label{limwu}
[t_*,\T)\times\realip\subset\cC\,\implies\,\lim_{s\to\T} w(s,\gamma_0(s))=0,
\end{align}
since $\ell\in C([0,T])$.
\vs{+4pt}

We now proceed in two steps.
\vs{+4pt}

\emph{Step $1$.} Here we prove that $\cS\cap([t,T)\times\realip)\neq\emptyset$ for all $t<T$. Assume the contrary, i.e.~there exists $t_*<T$ such that $[t_*,T)\times\realip\subset \cC$. Hence for any $t\in[t_*,\T)$ given and fixed we have $\tau_*(t,x)=\T-t$ $\prob$-a.s.~for all $x\in\realip$. In particular we take $x>\gamma_0(t)$. In order to obtain an upper bound for the value function we use the martingale property of \eqref{martw} and \eqref{tau0} to get
\begin{align}\label{limw00}
0\le w(t,x)=&\media\left[\int_0^{\tau_0}e^{-\int_0^{t+s}r(u)du}H(t+s,X^x_s)ds+w(t+\tau_0,X^x_{\tau_0})\right]\nonumber\\[+4pt]
\le&\,\media\left[\mathds{1}_{\{\tau_0\le \T-t-\eps\}}w(t+\tau_0,X^x_{\tau_0})\right]+\media\left[\mathds{1}_{\{\T-t-\eps<\tau_0< \T-t\}}w(t+\tau_0,X^x_{\tau_0})\right]\\[+4pt]
&-\delta_0 (T-t-\eps)\, c\,\prob(\tau_0> T-t-\eps)\nonumber
\end{align}
where $c>0$ is a uniform lower bound for the discount factor.

Since $X^x_{\tau_0}=\gamma_0(t+\tau_0)$ on $\{\tau_0<\T-t\}$ and $\gamma_0(\,\cdot\,)$ is bounded and continuous on $[t,\T-\eps]$, then, using that $w\in C([0,T]\times\realip)$, we can find $c_\eps\in[0,+\infty)$, independent of $x$ and such that
\begin{align}
\sup_{0\le s\le \T-t-\eps}w(t+s,\gamma_0(t+s))\le c_\eps.
\end{align}
On the other hand \eqref{limwu} implies that
\begin{align}
d_\eps:=\sup_{0\le s\le \eps}w(\T-s,\gamma_0(\T-s))<+\infty
\end{align}
so that \eqref{limw00} gives
\begin{align*}
0\le w(t,x)\le c_\eps\,\prob(\tau_0\le \T-t-\eps)+\prob(\tau_0>\T-t-\eps)\left(d_\eps\, -\delta_0 (T-t-\eps)\, c\,\right).
\end{align*}
Taking limits as $x\to\infty$ and recalling that $\tau_0(t,x)\to\T-t$ in probability (see \eqref{limprob}) we get
\begin{align*}
0\le w(t,x)\le d_\eps-\delta_0(T-t-\eps)\,c.
\end{align*}
Finally, letting $\eps\to 0$ and using \eqref{limwu} we obtain that $d_\eps\to0$ and
\begin{align*}
0\le w(t,x)\le -\delta_0(T-t-\eps)\,c,
\end{align*}
hence a contradiction. This means that $\cS\cap([t,T)\times\realip)\neq\emptyset$ for all $t<T$.
\vs{+4pt}

\emph{Step $2$.} Here we prove that $\cS\cap((t_1,t_2)\times\realip)\neq\emptyset$ for all $t_1<t_2$ in $[0,T]$. Let us argue by contradiction and assume that indeed $\cS\cap((t_1,t_2)\times\realip)=\emptyset$ for a given couple $t_1<t_2$ in $[0,T]$. With no loss of generality we can set
\[
t_2:=\sup\{t>t_1\,:\,\cS\cap((t_1,t)\times\realip)=\emptyset\}
\]
and we know from step 1 above that $t_2<T$. The idea is to show that indeed $t_2=t_1$, hence a contradiction.

Fix $t'\in[t_2,T)$ such that there exists $x'\in\realip$ with $(t',x')\in\cS$. Since $w_x\le 0$ then $\{t'\}\times[x',+\infty)\in\cS$ and we define
\begin{align}\label{tau'}
\tau':=\inf\{s\ge0\,:\,X^x_s\le x'\}.
\end{align}
As in \eqref{limw00} we use the martingale property of \eqref{martw} up to the stopping time $\zeta:=\tau_*\wedge\tau_0\wedge(t'-t_1)$, where here $\tau_*=\tau_*(t_1,x)$ and $\tau_0=\tau_0(t_1,x)$. In particular, noticing that $\prob(\tau_*\ge t_2-t_1)=1$ and therefore $\prob(\zeta\ge \tau_0\wedge(t_2-t_1))=1$, we get
\begin{align}\label{limw01}
0\le w(t_1,x)\le&\, -\delta_0(t_2-t_1)\,c\,\prob(\tau_0>t_2-t_1)\nonumber\\[+4pt]
&+\media\left[w(t_1+\tau_0,\gamma_0(t_1+\tau_0))\mathds{1}_{\{\tau_0<\tau_*\wedge(t'-t_1)\}}\right]\\[+3pt]
&+\media\left[w(t',X^x_{t'-t_1})\mathds{1}_{\{t'-t_1<\tau_*\wedge\tau_0\}}\right].\notag
\end{align}

On the event $\{\tau_0<\tau_*\wedge(t'-t_1)\}$ we have $\gamma_0(t_1+\tau_0)\le\sup_{t_1\le s \le t'}\gamma_0(s)\le \gamma'$ for some $0\le \gamma'<\infty$ that depends only on $t_1$ and $t'$. It follows that $w(t_1+\tau_0,X^x_{\tau_0})\le C$ for some $C>0$ that also depends only on $t_1$ and $t'$. Moreover $\{t'-t_1<\tau_*\wedge\tau_0\}=\{t'-t_1<\tau_*\wedge\tau_0\}\cap\{\tau'\le t'-t_1\}$ as otherwise the process would have crossed the vertical half-line $\{t'\}\times[x',+\infty)\in\cS$. For the same reason on the event $\{t'-t_1<\tau_*\wedge\tau_0\}$ one has $X^x_{t'-t_1}\le x'$ which then implies $w(t',X^x_{t'-t_1})\le C'$ for some $C'>0$ that only depends on $t_1$, $t'$ and $x'$.

Collecting these facts we have from \eqref{limw01}
\begin{align}\label{contr00}
0\le w(t_1,x)\le&\,C\,\prob(\tau_0\le t'-t_1)+ C'\,\prob(\tau'\le t'-t_1)-\delta_0(t_2-t_1)\,c\,\prob(\tau_0>t_2-t_1).
\end{align}
We take limits as $x\to\infty$. From the same argument as in \eqref{limprob} we get $\prob(\tau_0\le t'-t_1)\to 0$, $\prob(\tau_0>t_2 -t_1)\to 1$ and similarly
\begin{align*}
\prob(\tau'\le t'-t_1)=\prob\left(\inf_{0\le s\le t'-t_1}X^x_s\le x'\right)=\prob\left(\inf_{0\le s\le t'-t_1}X^1_s\le \frac{x'}{x}\right)\to 0.
\end{align*}
Thus \eqref{contr00} leads to a contradiction and it must be $t_2=t_1$.
\vs{+6pt}

Steps 1 and 2 above give $(iii)$.
\vs{+6pt}

Finally we can prove \eqref{limw0}. The limit is obvious for $t=\T$ since $w(\T,x)=0$, $x\in\realip$. It is also trivially true for all $t\in[0,T]$ such that $\cS\cap(\{t\}\times\realip)\neq\emptyset$. Therefore it remains to prove it for $t\in[0,T]$ such that $\cS\cap(\{t\}\times\realip)=\emptyset$.

Take one such $t\in[0,T]$ and fix $t'>t$ such that there exists $x'\in\realip$ with $(t',x')\in\cS$ (this must exist due to $(iii)$). Recall $\tau'$ as in \eqref{tau'}, then repeating the martingale argument and the estimates above we obtain
\begin{align*}
0\le w(t,x)\le&\,\media\left[w(t+\tau_0,\gamma_0(t+\tau_0))\mathds{1}_{\{\tau_0<\tau_*\wedge(t'-t)\}}\right]+\media\left[w(t',X^x_{t'-t})\mathds{1}_{\{t'-t<\tau_*\wedge\tau_0\}}\right]\\
\le&\,C\,\prob(\tau_0\le t'-t)+ C'\,\prob(\tau'\le t'-t).
\end{align*}
Taking limits as $x\to\infty$ \eqref{limw0} is easily verified.
\hfill$\square$

\subsection*{Proof of uniqueness in Theorem \ref{ThVolt1}} We give the proof only in the case of $\cS=\{(t,x)\,:\,x\le b(t)\}$ as the other case follows by the same arguments. Also, here we assume $\gamma(T)<+\infty$ for simplicity and notice that the argument below can be easily adapted for $\gamma(T)=+\infty$.

Let us assume there exists a continuous function $c:[0,T]\to \realip$ with $c(T)=\gamma(T)$, with $c(t)\le \gamma(t)$ for all $t\in[0,T]$ and such that $c$ solves
\begin{align}\label{inteq}
G(t,c(t))\!=\! \media &\, \bigg[ e^{-\int_{0}^{T-t}r(t+u)du}G(T,X^{c(t)}_{T-t})\notag\\
& +\!\int_{0}^{T-t}\!e^{-\int_{0}^{s}r(t+u)du}\!\left(\beta(t+s)X^{c(t)}_s \!-\!H(t+s,X^{c(t)}_s)\mathds{1}_{\{X^{c(t)}_s\le c(t+s)\}}\right) ds\bigg].
\end{align}
Then we define a function
\begin{align}
U^c(t,x)\!=\! \media &\, \bigg[ e^{-\int_{0}^{T-t}r(t+u)du}G(T,X^{x}_{T-t})\notag\\
& +\!\int_{0}^{T-t}\!e^{-\int_{0}^{s}r(t+u)du}\!\left(\beta(t+s)X^{x}_s \!-\!H(t+s,X^{x}_s)\mathds{1}_{\{X^{x}_s\le c(t+s)\}}\right) ds\bigg]
\end{align}
and notice that this is the analogue for $c$ of the value function $V$ in \eqref{reprV1}. Notice also that $U^c(T,x)=G(T,x)$ for $x\ge 0$.

Thanks to strong Markov property, it is not hard to show that the process $(\UU_s)_{s\in[0,T-t]}$ is a martingale, where
\begin{align}
\UU_s:=& e^{-\int_{0}^{s}r(t+u)du}U^c(t+s,X^{x}_{s})\notag\\
& +\!\int_{0}^{s}\!e^{-\int_{0}^{u}r(t+v)dv}\!\left(\beta(t+u)X^{x}_u \!-\!H(t+u,X^{x}_u)\mathds{1}_{\{X^{x}_u\le c(t+u)\}}\right) du.
\end{align}
Then the same argument also implies that $(\VV_s)_{s\in[0,T-t]}$ is a martingale as well, where
\begin{align}
\VV_s:=& e^{-\int_{0}^{s}r(t+u)du}V(t+s,X^{x}_{s})\notag\\
& +\!\int_{0}^{s}\!e^{-\int_{0}^{u}r(t+v)dv}\!\left(\beta(t+u)X^{x}_u \!-\!H(t+u,X^{x}_u)\mathds{1}_{\{X^{x}_u\le b(t+u)\}}\right) du.
\end{align}

Now we proceed in four steps as it is customary in the literature (see \cite{Pe05} and \cite{PS}).
\vs{+4pt}

\emph{Step 1}. First we show that $U^c(t,x)=G(t,x)$ for all $x\le c(t)$, $t\in[0,T]$. The statement is trivial for $(t,x)\in\{T\}\times\realip$ or for $x=c(t)$, since it follows by definition of $c(\cdot)$ and $U^c$. Let us now take $t<T$ and $x< c(t)$, define $\sigma_c:=\inf\{s\ge 0\,:\, X^x_s\ge c(t+s)\}\wedge(T-t)$ and use the martingale property of $\UU$ to obtain
\begin{align}
U^c(t,x)\!=\! \media &\, \bigg[ e^{-\int_{0}^{\sigma_c}r(t+u)du}U^c(t+\sigma_c,X^{x}_{\sigma_c})\notag\\
& +\!\int_{0}^{\sigma_c}\!e^{-\int_{0}^{s}r(t+u)du}\!\left(\beta(t+s)X^{x}_s \!-\!H(t+s,X^{x}_s)\right) ds\bigg].
\end{align}
Using that $U^c(t+\sigma_c,X^{x}_{\sigma_c})=G(t+\sigma_c,X^x_{\sigma_c})$, $\prob$-a.s., because $c$ solves \eqref{inteq} and $U^c(T,x)=G(T,x)$ we obtain
\begin{align}
U^c(t,x)\!=\! \media &\, \bigg[ e^{-\int_{0}^{\sigma_c}r(t+u)du}G(t+\sigma_c,X^{x}_{\sigma_c})\notag\\
& +\!\int_{0}^{\sigma_c}\!e^{-\int_{0}^{s}r(t+u)du}\!\left(\beta(t+s)X^{x}_s \!-\!H(t+s,X^{x}_s)\right) ds\bigg]=G(t,x),
\end{align}
where the final equality also uses that $\beta(t)x-H(t,x)=-(G_t+\cL G-r(\cdot)G)(t,x)$ and Dynkin's formula.
\vs{+4pt}

\emph{Step 2}. Now we show that $V(t,x)\ge U^c(t,x)$. The claim is trivial for $x\le c(t)$, $t\in[0,T)$ due to step 1. Similarly $U^c(T,x)=V(T,x)=G(T,x)$ for $x\in\realip$. Then, fix $t<T$, take $x>c(t)$ and denote $\tau_c:=\inf\{s\ge 0\,:\,X^x_s\le c(t+s)\}\wedge(T-t)$. Using the martingale property of $\UU$ and \eqref{inteq} we obtain
\begin{align}
U^c(t,x)\!=&\media \bigg[ e^{-\int_{0}^{\tau_c}r(t+u)du}U^c(t+\tau_c,X^{x}_{\tau_c})+\!\int_{0}^{\tau_c}\!e^{-\int_{0}^{s}r(t+u)du}\!\beta(t+s)X^{x}_s  ds\bigg]\notag\\
=&\media \bigg[ e^{-\int_{0}^{\tau_c}r(t+u)du}G(t+\tau_c,X^{x}_{\tau_c})+\!\int_{0}^{\tau_c}\!e^{-\int_{0}^{s}r(t+u)du}\!\beta(t+s)X^{x}_s  ds\bigg]\le V(t,x).
\end{align}
\vs{+4pt}

\emph{Step 3}. Here we prove that $c(t)\ge b(t)$, $t\in[0,T]$. Assume there is $t\in[0,T)$ such that $c(t)<b(t)$, take $x\le c(t)$ and denote $\sigma_b:=\inf\{s\ge 0\,:\, X^x_s\ge b(t+s)\}\wedge(T-t)$. Using now the martingale property of both $\UU$ and $\VV$ we obtain
\begin{align}
\label{V1}V(t,x)\!=\! \media &\, \bigg[ e^{-\int_{0}^{\sigma_b}r(t+u)du}V(t+\sigma_b,X^{x}_{\sigma_b})\notag\\
& +\!\int_{0}^{\sigma_b}\!e^{-\int_{0}^{s}r(t+u)du}\!\left(\beta(t+s)X^{x}_s \!-\!H(t+s,X^{x}_s)\right) ds\bigg]\\
\label{U1}U^c(t,x)\!=\! \media &\, \bigg[ e^{-\int_{0}^{\sigma_b}r(t+u)du}U^c(t+\sigma_b,X^{x}_{\sigma_b})\notag\\
& +\!\int_{0}^{\sigma_b}\!e^{-\int_{0}^{s}r(t+u)du}\!\left(\beta(t+s)X^{x}_s \!-\!H(t+s,X^{x}_s)\mathds{1}_{\{X^{x}_s\le c(t+s)\}}\right) ds\bigg].
\end{align}
Notice that $U^c(t,x)=V(t,x)$ because $x\le c(t)<b(t)$ and recall that $V(t+\sigma_b,X^{x}_{\sigma_b})\ge U^c(t+\sigma_b,X^{x}_{\sigma_b})$. Then, subtracting \eqref{U1} from \eqref{V1} we get
\begin{align}\label{2}
0\le \! \media &\, \bigg[\!\int_{0}^{\sigma_b}\!e^{-\int_{0}^{s}r(t+u)du}\!H(t+s,X^{x}_s)\mathds{1}_{\{X^{x}_s>c(t+s)\}} ds\bigg].
\end{align}
Since $H(t+s,X^{x}_s)<0$ for $s\le \sigma_b$ (recall that $H<0$ below $\gamma(\cdot)$) and since $\prob(X^{x}_s>c(t+s),\,\sigma_b>s)>0$ for any $s$ sufficiently small (thanks to continuity of $b$ and $c$), the inequality in \eqref{2} is a contradiction hence it cannot be $c(t)< b(t)$.
\vs{+4pt}

\emph{Step 4}. In this final step we show that $c(t)\le b(t)$ for $t\in[0,T]$, so that from step 3 we conclude $c(t)=b(t)$ for $t\in[0,T]$. Assume that there is $t\in[0,T)$ such that $c(t)>b(t)$ and take $x\in(b(t),c(t))$. Then letting $\tau_b=\inf\{s\ge 0\,:\,X^x\le b(t+s)\}\wedge(T-t)$ and using again the martingale property of both $\UU$ and $\VV$ we get
\begin{align}
\label{V2}V(t,x)\!=\! \media &\, \bigg[ e^{-\int_{0}^{\tau_b}r(t+u)du}V(t+\tau_b,X^{x}_{\tau_b})+\!\int_{0}^{\tau_b}\!e^{-\int_{0}^{s}r(t+u)du}\!\beta(t+s)X^{x}_s ds\bigg]\\
\label{U2}U^c(t,x)\!=\! \media &\, \bigg[ e^{-\int_{0}^{\tau_b}r(t+u)du}U^c(t+\tau_b,X^{x}_{\tau_b})\notag\\
& +\!\int_{0}^{\tau_b}\!e^{-\int_{0}^{s}r(t+u)du}\!\left(\beta(t+s)X^{x}_s \!-\!H(t+s,X^{x}_s)\mathds{1}_{\{X^{x}_s\le c(t+s)\}}\right) ds\bigg].
\end{align}
We notice that $V(t+\tau_b,X^{x}_{\tau_b})=U^c(t+\tau_b,X^{x}_{\tau_b})=G(t+\tau_b,X^{x}_{\tau_b})$ due to step 3 and step 1. Also we recall that $V(t,x)\ge U^c(t,x)$ due to step 2. Then subtracting \eqref{U2} from \eqref{V2} gives
\begin{align}\label{3}
0\le \media &\, \bigg[\!\int_{0}^{\tau_b}\!e^{-\int_{0}^{s}r(t+u)du}\!H(t+s,X^{x}_s)\mathds{1}_{\{X^{x}_s\le c(t+s)\}}ds\bigg].
\end{align}
Since $c(\cdot)\le \gamma(\cdot)$ and $c(t)<\gamma(t)$ then $H(t+s,X^{x}_s)\mathds{1}_{\{X^{x}_s\le c(t+s)\}}\le0$ and it is strictly negative for all $s\in[0,\tau_b)$ sufficiently small. Moreover, continuity of $c(\cdot)$ and $b(\cdot)$ imply that $\prob(X^{x}_s\le c(t+s),\,\tau_b>s)>0$ for all $s>0$ sufficiently small. Hence \eqref{3} gives a contradiction and it must be $c(t)\le b(t)$.
\hfill$\square$
\vspace{+15pt}

\noindent{\bf Acknowledgment}: This work was financially supported by Sapienza University of Rome, research project ``\emph{Polizze `Deferred Income Annuities' per la previdenza complementare: un modello stocastico di valutazione}'', grant no.~RP11615500B7B502. T.~De Angelis was also partially supported by the EPSRC grant EP/R021201/1, ``\emph{A probabilistic toolkit to study regularity of free boundaries in stochastic optimal control}''.

Finally, we would like to thank three anonymous referees, whose valuable comments helped improving the quality of our paper.

\end{document}